\documentclass{article}

\PassOptionsToPackage{numbers, compress}{natbib}
 \usepackage[preprint]{neurips_2026}

\usepackage[utf8]{inputenc} 
\usepackage[T1]{fontenc}    
\usepackage{hyperref}       
\usepackage{url}            
\usepackage{booktabs}       
\usepackage{amsfonts}       
\usepackage{nicefrac}       
\usepackage{microtype}      
\usepackage{xcolor}         

\usepackage{wrapfig}
\usepackage{graphicx}
\usepackage{mathtools}
\usepackage{amsmath, amssymb, amsthm}
\usepackage{bm}
\usepackage{enumitem}
\usepackage{multirow}
\usepackage{makecell}
\usepackage{siunitx}
\usepackage{adjustbox}
\usepackage{tabularx}
\usepackage{diagbox}
\usepackage[most]{tcolorbox}
\usepackage{anyfontsize}
\usepackage{hyperref}       
\usepackage{url}            
\usepackage{booktabs}       
\usepackage{amsfonts}       
\usepackage{nicefrac}       
\usepackage{microtype}      
\usepackage{xcolor}         
\usepackage{natbib, mathtools}
\usepackage{multirow}
\usepackage{graphicx}
\usepackage[table]{xcolor}
\usepackage{array}
\usepackage{siunitx}
\usepackage{ragged2e}
\usepackage{listings}
\newtheorem{theorem}{Theorem}

\newtheorem{remark}{Remark}

\newtheorem{assumption}{Assumption}

\usepackage{caption}
\title{RLSpoofer: A Lightweight Evaluator for LLM Watermark Spoofing Resilience.}

\author{
Hanbo Huang\textsuperscript{1}
\quad
Xuan Gong\textsuperscript{1}
\quad
Yiran Zhang\textsuperscript{1}
\quad
Hao Zheng\textsuperscript{1}
\quad
Shiyu Liang\textsuperscript{1~$\dagger$}
\\
\textsuperscript{1}Shanghai Jiao Tong University\\
{\tt\small \{hhuang417, lsy18602808513\}@sjtu.edu.cn}
}

\begin{document}

\def\thefootnote{$\dagger$}\footnotetext{Corresponding author.}

\maketitle

\begin{abstract}
Large language model (LLM) watermarking has emerged as a promising approach for detecting and attributing AI-generated text, yet its robustness to black-box spoofing remains insufficiently evaluated. Existing evaluation methods often demand extensive datasets and white-box access to algorithmic internals, limiting their practical applicability. In this paper, we study watermark resilience against spoofing fundamentally from a distributional perspective. We first establish a \textit{local capacity bottleneck}, which theoretically characterizes the probability mass that can be reallocated under KL-bounded local updates while preserving semantic fidelity. Building on this, we propose RLSpoofer, a reinforcement learning-based black-box spoofing attack that requires only 100 human-watermarked paraphrase training pairs and zero access to the watermarking internals or detectors. Despite weak supervision, it empowers a 4B model to achieve a 62.0\% spoof success rate with minimal semantic shift on PF-marked texts, dwarfing the 6\% of baseline models trained on up to 10,000 samples. Our findings expose the fragile spoofing resistance of current LLM watermarking paradigms, providing a lightweight evaluation framework and stressing the urgent need for more robust schemes.

\end{abstract}

\section{Introduction}
With the rapid advancement and increasing availability of large language models (LLMs), they are being widely applied across diverse applications to generate fluent and human-like content~\citep{yang2025qwen3,grattafiori2024llama}. However, this widespread use raises major concerns about model misuse, including the generation of misinformation~\cite{chen2023can}, copyright violations~\cite{xu2025copyright}, data contamination~\cite{shumailov2023curse}, and academic dishonesty~\cite{cotton2024chatting}. As a safeguard, text watermarking has become an important defense. By subtly embedding statistical signals into model outputs, watermarking enables reliable AI text detection and source tracking while maintaining the quality of the text~\citep{kirchenbauer2023watermarkKGW, zhao2024permutePF}.

Most existing watermarking schemes predominantly operate on a generate-and-detect paradigm, utilizing detection algorithms to identify hidden patterns that differentiate AI-generated text from human authorship~\citep{lu2024entropyEWD,lee2023wroteSWEET}. However, this paradigm inadvertently creates an avenue for watermark spoofing. Malicious actors can query a watermarked LLM to gather text samples, using this data to construct a surrogate model capable of approximating the underlying watermarking rules~\cite{jovanovic2024watermark}. By successfully forging the watermark, such attacks pose a fundamental threat to the reliability of existing watermarking systems.

To rigorously assess the spoofing resistance of LLM watermarking schemes, researchers have investigated various attack paradigms. However, existing methods typically suffer from at least one of three major limitations: (1) reliance on partial algorithmic knowledge or detector API access~\cite{kirchenbauer2023watermarkKGW,chen2024mark}; (2) the need for large training corpora~\cite{an2025ditto,gu2023learnabilitydistilled} (up to 10k); or (3) an inability to generate semantically coherent text that seamlessly matches the original context~\cite{pang2024attacking,gloaguen2024discovering}. These limitations reduce their practicality as evaluation tools for watermark spoofing resilience. As LLMs continue to advance, there is an urgent need for a practical, generalizable, and data-efficient methodology for benchmarking the robustness of watermarking schemes.

In this paper, we investigate the resilience of LLM watermarks to spoofing attacks. In the black-box setting, direct instance-level spoofing is intractable because the adversary has no access to the watermark details or the detector. Motivated by recent evidence that watermarked text exhibits systematic distributional discrepancy from human-written text~\cite{huang2025rlcracker}, even for distortion-free watermarking schemes~\cite{liu2024can}, we take a \textit{distributional view} of watermark spoofing. Specifically, the attacker seeks to shift the paraphraser’s output distribution away from the human-like distribution and toward the watermarked distribution.

\begin{figure}
    \centering
    \includegraphics[width=\linewidth]{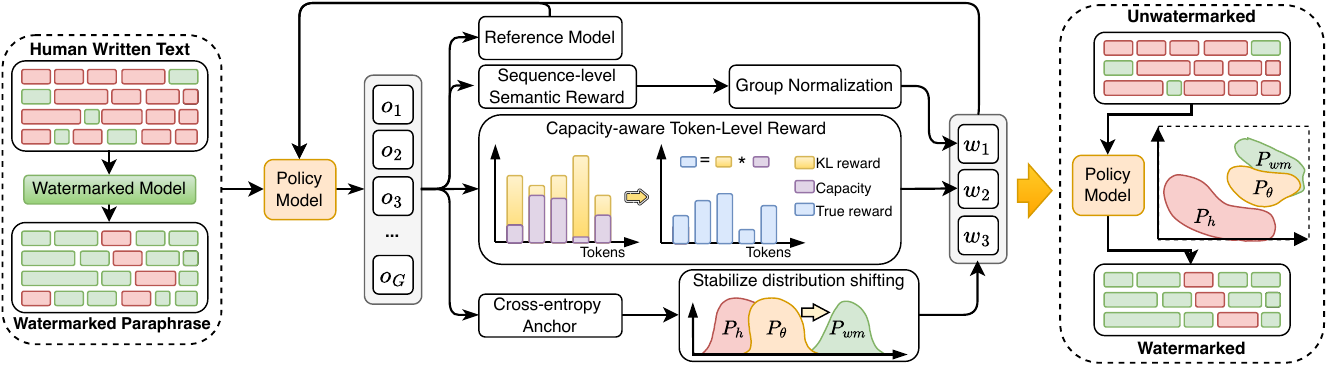}
    \caption{Overview of RLSpoofer. Given human--watermarked rewrite pairs, RLSpoofer jointly optimizes sequence-level semantic rewards, capacity-aware token-level rewards, and a cross-entropy anchor, shifting the policy distribution $P_{\theta}$ from the human-like distribution $P_h$ toward the watermarked distribution $P_{wm}$, thereby effectively spoofing the watermark.}
    \label{fig:overview_spoofer}
    \vspace{-0.4cm}
\end{figure}

Leveraging the autoregressive structure of LLM generation, we establish a \textit{local capacity bottleneck} that characterizes how much token-level probability mass can be reallocated under KL-bounded local updates while preserving semantic fidelity. Building on this insight, we propose {RLSpoofer}, a reinforcement-learning-based black-box attack that requires no access to the watermarking internals or the detector. As illustrated in Figure~\ref{fig:overview_spoofer}, RLSpoofer jointly optimizes semantic fidelity, capacity-aware token-level reward, and a cross-entropy anchor, thereby shifting the model outputs toward the target watermarked distribution. Using only \textbf{100 training pairs}, RLSpoofer enables a compact 4B model to achieve a \textbf{62.0\%} spoof success rate against PF-Watermark, substantially surpassing the \textit{6.0\%} of baseline trained on \textit{10,000 examples}. Extensive experiments across five attacker models and six watermarking schemes show that RLSpoofer provides a practical, sample-efficient stress test for spoofing resilience, exposing critical vulnerabilities in current watermarking defenses. Our contributions are as follows:
\begin{itemize}[leftmargin=*, itemsep=1pt, parsep=0pt, topsep=0pt]
    \item We establish a \textit{local capacity bottleneck} that characterizes the token-level bottleneck on probability mass reallocation under KL-bounded local updates while preserving semantic integrity. (Sec.~\ref{sec:local_capacity_bottleneck}.)
    \item We propose RLSpoofer, an effective and sample-efficient black-box attack. It successfully spoofs watermarks with high semantic fidelity, requiring only limited human-watermarked rewrite pairs and zero access to the watermarking scheme or detector. (Sec.~\ref{sec:rlspoofer}.)
    \item We conduct extensive experiments across five architectures and six watermarking schemes from both logit-based and sampling-based families. The results demonstrate the broad effectiveness of RLSpoofer, exposing critical vulnerabilities in current watermark defense. (Sec.~\ref{sec:experiments}.)
\end{itemize}

\section{Preliminaries}
\label{sec:preliminaries}

\textbf{LLM Paraphraser and Logit-based Watermark.}
Let \(\mathcal{V}\) be a finite vocabulary and \(\mathcal{V}^*\) denote the space of all finite token sequences. For a source input \(\mathbf{X} \in \mathcal{V}^*\), an autoregressive paraphraser \(\pi_\theta\) induces a probability distribution over generated output sequences \(\mathbf{X}' = (x'_1, \dots, x'_{|\mathbf{X}'|}) \in \mathcal{V}^*\): $P_\theta(\mathbf{X}' \mid \mathbf{X}) := \prod_{t=1}^{|\mathbf{X}'|} \pi_\theta(x'_t \mid x'_{<t}, \mathbf{X}).$ Independent of the paraphrasing model, a watermark is embedded into generated text using a secret key \(s\) and a target green-token rate \(q \in (0,1)\). At each generation step, key \(s\) defines a prefix-dependent pseudorandom subset of the vocabulary with expected size \(q|\mathcal{V}|\), whose tokens are softly upweighted during sampling. Given a sequence \(\mathbf{X}'\), the detector $f$ computes a score \(f(\mathbf{X}', s)\) under the key \(s\) and flags the sequence watermarked whenever \(f(\mathbf{X}', s) > \delta\), where \(\delta\) is a fixed detection threshold.

\textbf{Threat Model: Semantics-Preserving Paraphrasing.}  
We consider a black-box adversary aiming to spoof the watermark into an unwatermarked sequence \(\mathbf{X}\) via paraphrasing. The adversary is defined by:
(1) \textbf{\textit{Objective}}: Generate an output \(\mathbf{X}'\) that preserves the semantic meaning of \(\mathbf{X}\) while triggering the watermark detector.
(2) \textbf{\textit{Knowledge}}: The adversary lacks access to the secret key \(s\), the detector function \(f\), and the threshold \(\delta\), but can observe watermarked text by querying the generator.
(3) \textbf{\textit{Capability}}: The adversary can construct a training set by querying the watermarked model for paraphrases of given texts, and can subsequently fine-tune a paraphrasing model $\pi_\theta$ without requiring access to the detector.

\textbf{Adversarial Objective.}
Let \(d: \mathcal{V}^* \times \mathcal{V}^* \to [0,\infty)\) be a distance metric where smaller values indicate higher semantic similarity. For a tolerance \(\varepsilon>0\), we define the instance-level spoof success rate as: \(\operatorname{SSR}_{\mathbf{X}}(\theta) := P_\theta(\{ \mathbf{X}' \in \mathcal{V}^* : d(\mathbf{X}', \mathbf{X}) \le \varepsilon, \, f(\mathbf{X}', s) > \delta \} \mid \mathbf{X})\). The adversary seeks to maximize this rate. However, without access to \(f\) and \(s\), direct optimization is intractable.


\textbf{Distributional Surrogate.}
Prior work suggests a distributional discrepancy between watermarked and human-written text~\cite{huang2025rlcracker,liu2024can}. Since the detector function $f$ and the secret key $s$ are unavailable in the black-box setting, directly optimizing the instance-level spoof success rate is intractable. We therefore adopt a distributional surrogate objective that favors outputs relatively more likely under a watermarked distribution than under a human-like distribution, while enforcing semantic fidelity. Let $P_h(\cdot \mid \mathbf{X})$ and $P_{wm}(\cdot \mid \mathbf{X})$ denote the human-like and watermarked distributions, respectively. We consider the following surrogate objective based on the Kullback--Leibler (KL) divergences $D_{\mathrm{KL}}$:
\begin{equation*}
\max_\theta \;
D_{\mathrm{KL}}\!\big(P_\theta(\cdot \mid \mathbf{X}) \,\|\, P_h(\cdot \mid \mathbf{X})\big)
-
D_{\mathrm{KL}}\!\big(P_\theta(\cdot \mid \mathbf{X}) \,\|\, P_{wm}(\cdot \mid \mathbf{X})\big).
\end{equation*}
Under the mild regularity conditions in Appendix~\ref{app:assumption}, the KL-difference objective admits an expected log-likelihood-ratio form. Imposing the semantic-fidelity constraint $P_\theta(d(X',X) < \varepsilon \mid X) \ge 1-\rho$ for $\rho \in [0,1)$ yields the following constrained surrogate objective:
\begin{equation}
\max_\theta \;
\mathbb{E}_{\mathbf{X}' \sim P_\theta(\cdot \mid \mathbf{X})}
\left[
\log \frac{P_{wm}(\mathbf{X}' \mid \mathbf{X})}{P_h(\mathbf{X}' \mid \mathbf{X})}
\right]
\quad
\text{s.t.}\quad
P_\theta\!\bigl(d(\mathbf X',\mathbf X)<\varepsilon \mid \mathbf X\bigr)\geq 1-\rho.
\label{eq:J_clean}
\end{equation}
This formulation yields a tractable surrogate objective that prefers semantically faithful outputs with higher relative likelihood under the watermarked reference than under the human-like reference.


\section{Methods}
\label{sec:methods}

\subsection{Local Capacity Bottleneck for Semantics-Preserving KL-Bounded Redistribution}
\label{sec:local_capacity_bottleneck}

To optimize the sequence-level objective in Eq.~\ref{eq:J_clean}, we decompose it into token-level steps. Let the generation history at step \(t\) be \(h_t:=(x'_{<t},\mathbf X)\). Since both reference distributions are autoregressive, $\log\frac{P_{wm}(\mathbf X'\mid\mathbf X)}{P_h(\mathbf X'\mid\mathbf X)}=\sum_{t=1}^{|\mathbf X'|}\log\frac{P_{wm}(x'_t\mid h_t)}{P_h(x'_t\mid h_t)}.$
Accordingly, the adversarial objective becomes
\begin{equation}
\max_{\theta}\;
\mathbb E_{\mathbf X'\sim P_\theta(\cdot\mid\mathbf X)}
\sum_{t=1}^{|\mathbf X'|}
\left[
\log\frac{P_{wm}(x'_t\mid h_t)}{P_h(x'_t\mid h_t)}
\right]
\quad
\text{s.t.}
\quad
P_\theta\!\bigl(d(\mathbf X',\mathbf X)<\varepsilon \mid \mathbf X\bigr)\geq 1-\rho.
\label{eq:llr_decomposition_clean}
\end{equation}
This decomposition identifies a token-level surrogate signal. To localize the sequence-level semantic constraint, we consider KL-bounded local updates around the human-like reference $P_h(\cdot \mid h_t)$, using such neighborhoods as a tractable proxy for compatible modifications. This leads to the following local capacity question: \textit{under a fixed KL budget, how much probability mass can be reassigned?}

Formally, for a fixed history \(h_t\), we consider a local update from \(P_h(\cdot\mid h_t)\) to \(\pi_\theta(\cdot\mid h_t)\) under a KL constraint, and bound the resulting increase of probability mass on an arbitrary subset of tokens. Let $d_{\mathrm{kl}}(p\|q):=p\log\frac{p}{q}+(1-p)\log\frac{1-p}{1-q},
\ p,q\in[0,1],$ with the standard conventions \(0\log 0:=0\) and \(b\log(b/0):=+\infty\) for \(b>0\). We then establish the following theorem.
\begin{theorem}[Local capacity characterization]
\label{thm:capacity_clean}
Fix a history \(h_t\). For any next-token distribution \(\pi_\theta(\cdot\mid h_t)\) and any subset \(A\subseteq\mathcal V\), we have
\begin{equation}
\resizebox{0.93\linewidth}{!}{$
\pi_\theta(A\mid h_t)-P_h(A\mid h_t)
\le
\sup\Bigl\{
\lambda\in[0,1]:
d_{\mathrm{kl}}\!\Bigl(\lambda \,\Big\|\, 1-\max_{x\in\mathcal V}P_h(x\mid h_t)\Bigr)
\le
D_{\mathrm{KL}}\!\bigl(\pi_\theta(\cdot\mid h_t)\,\|\,P_h(\cdot\mid h_t)\bigr)
\Bigr\}.
$}
\end{equation}
Moreover, for any fixed \(C\ge 0\),
\begin{equation}
\lim_{1-\max_{x\in\mathcal V}P_h(x\mid h_t)\to 0}
\sup_{\pi_\theta:\,D_{\mathrm{KL}}(\pi_\theta(\cdot\mid h_t)\,\|\,P_h(\cdot\mid h_t))\le C}
\sup_{A\subseteq\mathcal V}
\bigl(\pi_\theta(A\mid h_t)-P_h(A\mid h_t)\bigr)=0.
\label{eq:capacity_vanish_clean_rewrite}
\end{equation}
\end{theorem}

\begin{remark}
The proof is provided in Appendix~\ref{app:proof_capacity_clean}. Theorem~\ref{thm:capacity_clean} characterizes a local bottleneck for KL-bounded redistribution around the human-like reference: when the next-token distribution is highly concentrated, only limited probability mass can be reassigned under a fixed KL budget. While purely distributional, this provides a tractable local proxy for compatible modification. In particular, taking \(A=\{x\}\) recovers the token-level case.
\end{remark}
Let $x_t^*\in\arg\max_{x\in\mathcal V}P_h(x\mid h_t)$, and define $c_t:=1-P_h(x_t^*\mid h_t)$. We refer to \(c_t\) as the \emph{local capacity mass}. By Theorem~\ref{thm:capacity_clean}, \(c_t\) quantifies the local bottleneck for semantics-consistent redistribution: a more concentrated human-like distribution leaves less probability mass available for reassignment under a fixed KL budget.

Moreover, if a local update preserves the dominant human-like continuation, i.e., $\pi_\theta(x_t^*\mid h_t)=P_h(x_t^*\mid h_t),$ then the local surrogate gain factorizes as
\begin{equation}
\resizebox{0.93\linewidth}{!}{$
\sum_{x\in\mathcal V}
\bigl(\pi_\theta(x\mid h_t)-P_h(x\mid h_t)\bigr)
\log\frac{P_{wm}(x\mid h_t)}{P_h(x\mid h_t)}
=
c_t
\sum_{x\neq x_t^*}
\frac{\pi_\theta(x\mid h_t)-P_h(x\mid h_t)}{c_t}
\log\frac{P_{wm}(x\mid h_t)}{P_h(x\mid h_t)}.
$}
\label{eq:capacity_factorization}
\end{equation}
Thus, $c_t$ serves as a local indicator of both the feasible amount of redistribution and the scale of the token-level surrogate gain. This suggests that token positions with larger $c_t$ may admit more flexible redistribution, thereby motivating a token-dependent weighting scheme.


\subsection{RLSpoofer: Capacity-Aware RL Attack for Watermark Spoofing }
\label{sec:rlspoofer}

\textbf{From surrogate objective to policy optimization.}
Equation~\ref{eq:J_clean} defines an expected sequence-level objective under the paraphraser-induced distribution over discrete generations, which naturally admits an episodic Markov decision process formulation with \(\pi_\theta\) serve as the policy. Motivated by this view, we propose RLSpoofer, a GRPO-based~\cite{shao2024deepseekmath} watermark spoofing attack that requires only limited sample pairs \((\mathbf{X}, \mathbf{X}'_{wm})\), where \(\mathbf{X}\) is a human-written text and \(\mathbf{X}'_{wm}\) is its semantic-preserving paraphrase generated by a watermarked model. Using such pairs, the attack steers the paraphrasing distribution \(P_\theta(\cdot \mid \mathbf{X})\) toward the watermarked distribution \(P_{wm}(\cdot \mid \mathbf{X})\) and away from the human-written distribution \(P_h(\cdot \mid \mathbf{X})\), thereby encouraging watermark-spoofed generations.

\textbf{Approximating the target distributions.}
Since the ground-truth watermarked and human-like distributions are not directly accessible, we approximate both using a \textit{lightweight} reference model \(\pi_{\mathrm{ref}}\). As LLMs are trained to model human text~\cite{grattafiori2024llama}, we approximate \(P_h\) by querying the \(\pi_{\mathrm{ref}}\) conditioned on the original input \(\mathbf{X}\), namely $P_h(x'_t \mid x'_{<t}, \mathbf{X}) \approx \pi_{\mathrm{ref}}(x'_t \mid x'_{<t}, \mathbf{X}).$
Similarly, motivated by the observation that weak rewriting often preserves both semantics and watermark patterns~\cite{kirchenbauer2023watermarkKGW}, we approximate the watermarked distribution by conditioning the same reference model on the watermarked instance \(\mathbf{X}'_{wm}\), i.e., $P_{wm}(x'_t \mid x'_{<t}, \mathbf{X}) \approx \pi_{\mathrm{ref}}(x'_t | x'_{<t}, \mathbf{X}'_{wm}).$ Intuitively, conditioning on the semantics-preserving rewrite $X'_{wm}$ provides a tractable proxy for the watermarked distribution's lexical and semantic preferences, yielding a tractable local approximation of $P_{wm}$.

\textbf{Reward design.}
Based on Eq.~\ref{eq:llr_decomposition_clean}, the token-level spoofing signal is naturally given by $\log\frac{P_{wm}(x'_t\mid h_t)}{P_h(x'_t\mid h_t)}.$ Moreover, Theorem~\ref{thm:capacity_clean} together with Eq.~\ref{eq:capacity_factorization} suggest that, under semantics-consistent local updates around \(P_h(\cdot\mid h_t)\), the attainable improvement of this signal is modulated by the local capacity mass \(c_t\). Motivated by this, we define the reward for a sampled token \(x'_t\sim\pi_\theta(\cdot\mid h_t)\) as $r_t:=c_t\log\frac{P_{wm}(x'_t\mid h_t)}{P_h(x'_t\mid h_t)}$, where $c_t = 1-\max_{x\in\mathcal V}P_h(x\mid h_t) \approx 1-\max_{x\in\mathcal V}\pi_{\mathrm{ref}}(x\mid x'_{<t}, \mathbf{X}).$ The quantity \(c_t\) represents the human-plausible mass outside the dominant continuation, thereby capturing the local room for lexical redistribution under semantic fidelity. As a result, \(r_t\) rewards tokens preferred by the surrogate watermarked distribution while attenuating updates at positions where the human-like distribution is already sharply concentrated.

However, token-level redistribution alone is insufficient, since a successful attack must also preserve semantics at the sequence level. We therefore introduce a sequence-level reward \(A\) based on the P-SP score~\citep{wieting2021paraphrastic}, followed by a sigmoid transformation to improve gradient sensitivity. Specifically, we compute the semantic similarity between the generated paraphrase and each of \(\mathbf{X}\) and \(\mathbf{X}'_{wm}\), and take the minimum as \(A\). This conservative design preserves the meaning of the original text while keeping the policy close to the watermarked rewrite, thereby improving the quality of the surrogate target for watermark-favored generation. We further normalize \(A\) within each rollout group to obtain \(\hat{A}_i=\frac{A_i-\mathrm{mean}(A)}{\mathrm{std}(A)}\), which stabilizes training and improves semantic consistency across samples.

\textbf{Training objective.}
Recent work has shown that distillation-based objectives can be effective for watermark spoofing~\cite{gu2023learnabilitydistilled,an2025ditto}. Accordingly, in addition to the reward design above, RLSpoofer also incorporates a cross-entropy term to anchor optimization toward the target watermarked distribution. Concretely, the attack policy \(\pi_\theta\) is iteratively updated by sampling a set of outputs \(\{x'_1, \dots, x'_G\}\) from the previous policy \(\pi_{\theta_{\mathrm{old}}}\), maximizing the following training objective:
\begin{equation}
\resizebox{0.85\linewidth}{!}{$
  \begin{aligned}
    \mathcal{J}(\theta) \approx\;
    \mathbb{E}_{\{x'_i\}\sim\pi_{\theta_{\mathrm{old}}}}
    \frac{1}{G} \sum_{i=1}^G \frac{1}{|x'_i|} \sum_{t=1}^{|x'_i|} &
    \Biggl[
      \frac{\pi_\theta(x'_{i,t}\mid \mathbf{X}, x'_{i,<t})}
           {\pi_{\theta_{\mathrm{old}}}(x'_{i,t}\mid \mathbf{X}, x'_{i,<t})}
      \bigl(w_1 \hat{A}_i + w_2 r_{i,t}\bigr) \\
    &\qquad
      - \beta D_{\mathrm{KL}}\bigl[\pi_\theta \,\|\, \pi_{\mathrm{ref}}\bigr]
    \Biggr]
    - w_3\,\mathcal{L}_{\mathrm{CE}}\bigl(\pi_\theta, \{(\mathbf{X}, \mathbf{X}'_{wm})\}\bigr), 
  \end{aligned}$}
\end{equation}
where \(\pi_{\mathrm{ref}}\) denotes the reference model, and \(w_1, w_2, w_3\) control the contributions of the respective components. The KL term regularizes the policy against the reference model at the token level, while \(\mathcal{L}_{\mathrm{CE}}\) denotes the teacher-forced cross-entropy loss for predicting \(\mathbf{X}'_{wm}\) conditioned on \(\mathbf{X}\).

\section{Experiments}
\label{sec:experiments}

\subsection{Experimental Setup}
We provide detailed experimental setups and configurations in Appendix~\ref{app:experimental_setup}.

\textbf{Victim Models and Attackers.} We use Llama3.1-8B-Instruct~\citep{grattafiori2024llama} as the victim model to generate watermarked paraphrases.  For attackers, we consider five lightweight models of varying sizes from three well-known model families: Qwen3-0.6B, Qwen3-1.7B, Qwen3-4B~\citep{yang2025qwen3}, Qwen2.5-3B-Instruct~\citep{qwen2025qwen25technicalreport}, and Llama3.2-3B-Instruct.

\textbf{Watermarking Schemes.} We evaluate six watermarking schemes drawn from both logit-based and sampling-based families. Specifically, the logit-based methods include EWD~\citep{lu2024entropyEWD}, SWEET~\citep{lee2023wroteSWEET}, KGW~\citep{kirchenbauer2023watermarkKGW}, and Unigram~\citep{zhao2023provableUnigram}. To capture recent advances, we additionally consider two sampling-based distortion-free methods: semantic watermark PMark~\citep{huo2025pmark} and the cryptographic PF-Watermark~\citep{zhao2024permutePF}. We implement PMark based on its official codebase, whereas the remaining schemes are implemented and detected using the widely adopted MarkLLM toolkit~\citep{pan2024markllm}.

\textbf{Datasets.}
We construct the training set from C4-RealNewslike subset~\cite{raffel2020exploring} and the test set from four datasets: Reddit WritingPrompts~\citep{verma2024ghostbuster}, LFQA~\citep{krishna2023paraphrasing}, and the BookReport and FakeNews subsets of MMW~\citep{piet2025markmywords}, following~\cite{jovanovic2024watermark}. Specifically, we prompt Qwen3-8B to generate human-like unwatermarked texts and use watermarked Llama3.1-8B-Instruct to rewrite them into semantically preserved watermarked responses. For each watermarking scheme, we use 100 unwatermarked--watermarked response pairs for training and 400 unwatermarked test samples drawn evenly from the four datasets. We standardize the length of all training and test samples to 500 tokens.

\textbf{Spoof Methods.} Under the black-box threat model, in addition to RLSpoofer, we consider three spoofing baselines to further reveal vulnerabilities in watermarking schemes. \textbf{Distill}~\cite{gu2023learnabilitydistilled} distills the distribution of the watermarked model using 10{,}000 huamn--watermarked rewritten pairs. \textbf{DITTO}~\cite{an2025ditto} further exploits the distilled model by analyzing its output distribution and reproducing the statistical preferences induced by the target watermark. Additionally, \textbf{DPO}~\cite{diaa2024optimizing} casts watermark spoofing as a preference alignment and optimizes the attacker on 7{,}000 samples preference pairs.

\textbf{Implementation Details of RLSpoofer.} RLSpoofer utilizes a reparameterized sigmoid scaling function for the semantic reward, with a threshold of 0.85 to separate positive and negative rewards, thereby encouraging higher-quality rephrasings. Training on 100 samples with a batch size of 48 and group size $G=12$ converges efficiently, completing in approximately 1.5 hours for Qwen3-4B and 45 minutes for Qwen3-0.6B on four NVIDIA Pro 6000 GPUs. In the experiments, we use the attack model itself as the reference model to generate the surrogate watermark distribution.

\textbf{Metric.} 
We primarily evaluate spoofing effectiveness using Spoof Success Rate (SSR), defined as the proportion of rephrased texts detected as watermarked under a semantic similarity constraint (P-SP > 0.7, following prior work~\citep{jovanovic2024watermark}). We also report Spoof Rate (SR), which omits the semantic constraint. Rephrasing quality is further assessed using perplexity and GPT-as-a-Judge (GPTS)~\citep{SIRA}.

\newpage
\subsection{Main Results}
\label{sec:main_Results_main}

In this subsection, we evaluate the spoofing resilience of different watermarking schemes across four attacks, with results shown in Table~\ref{tab:spoofmain}. Details of rephrase quality are provided in Appendix~\ref{app:main_results}.

\begin{table}[t]
  \centering
  \caption{Spoof Success Rate (SSR, \%) and P-SP scores across models and watermarking schemes. \textbf{Bold} indicates the \textit{best} SSR for each model. Higher SSR indicates stronger spoofing performance.}
  \label{tab:spoofmain}
  \small
  \setlength{\tabcolsep}{3.3pt} 
  \renewcommand{\arraystretch}{0.9} 
  \begin{tabular}{@{}l l cccccccccccccc}
    \toprule
    & & \multicolumn{2}{c}{\textbf{EWD}} & \multicolumn{2}{c}{\textbf{SWEET}} & \multicolumn{2}{c}{\textbf{KGW}} & \multicolumn{2}{c}{\textbf{Unigram}} & \multicolumn{2}{c}{\textbf{PF}} & \multicolumn{2}{c}{\textbf{PMark}} \\
    \cmidrule(lr){3-4} \cmidrule(lr){5-6} \cmidrule(lr){7-8} \cmidrule(lr){9-10} \cmidrule(lr){11-12} \cmidrule(lr){13-14} 
    \makecell[l]{\textbf{Models}} & \makecell[l]{\textbf{Methods}}
      & {\small\textbf{SSR}} & {\small\textbf{P-SP}} & {\small\textbf{SSR}} & {\small\textbf{P-SP}} & {\small\textbf{SSR}} & {\small\textbf{P-SP}} & {\small\textbf{SSR}} & {\small\textbf{P-SP}} & {\small\textbf{SSR}} & {\small\textbf{P-SP}} & {\small\textbf{SSR}} & {\small\textbf{P-SP}} \\
    \midrule
    \multirow{4}{*}{Qwen3-0.6B}
      & Distill        & 42.3 & 0.75 & 20.0 & 0.76 & 35.8 & 0.68 & 13.8 & 0.84 & 6.50 & 0.97 & 20.0 & 0.90 \\
      & DITTO        & 7.50 & 0.43 & 6.75 & 0.41 & 1.00 & 0.33 & 0.25 & 0.32 & 5.50 & 0.79 & 11.8 & 0.63  \\
      & DPO        & 0.25 & 0.57 & 0.00 & 0.76 & 1.00 & 0.66 & 0.25 & 0.32 & 2.50 & 0.57 & 6.25 & 0.63  \\
      & \textbf{RLSpoofer}   & \textbf{54.3} & 0.73 & \textbf{50.5} & 0.79 & \textbf{52.0} & 0.72 & \textbf{49.5} & 0.70 & \textbf{33.3} & 0.66 & \textbf{29.5} & 0.92  \\
    \midrule
    
    \multirow{4}{*}{Qwen3-1.7B}
      & Distill        & 43.8 & 0.80 & 26.8 & 0.79 & 44.5 & 0.75 & 19.5 & 0.81 & 7.00 & 0.96 & 20.3 & 0.90 \\
      & DITTO        & 21.5 & 0.53 & 29.0 & 0.56 & 13.8 & 0.52 & 1.50 & 0.36 & 7.50 & 0.86 & 16.5 & 0.68  \\
      & DPO        & 0.25 & 0.94 & 0.00 & 0.84 & 1.00 & 0.96 & 0.50 & 0.87 & 4.25 & 0.58 & 22.5 & 0.93  \\
      & \textbf{RLSpoofer}   & \textbf{53.5} & 0.76 & \textbf{52.0} & 0.71 & \textbf{52.0} & 0.71 & \textbf{54.8} & 0.73& \textbf{29.0} & 0.73 & \textbf{29.5} & 0.90  \\
    \midrule
    
    \multirow{4}{*}{Qwen3-4B}
      & Distill        & 51.3 & 0.81 & 37.3 & 0.82 & 57.0 & 0.79 & 28.0 & 0.80 & 6.00 & 0.96 & 21.5 & 0.92 \\
      & DITTO        & 56.0 & 0.68 & 43.3 & 0.66 & 36.5 & 0.61 & 2.25 & 0.39 & 3.50 & 0.88 & 16.5 & 0.71  \\
      & DPO        & 0.25 & 0.78 & 0.00 & 0.78 & 1.00 & 0.93 & 0.75 & 0.59 & 5.25 & 0.88 & 17.5 & 0.68  \\
      & \textbf{RLSpoofer}   & \textbf{56.5} & 0.73 & \textbf{52.3} & 0.75 & \textbf{58.0} & 0.75 & \textbf{54.8} & 0.74 & \textbf{62.0} & 0.77 & \textbf{36.3} & 0.91  \\
    \midrule
    
    \multirow{4}{*}{\makecell[l]{Qwen2.5-3B\\-Instruct}}
      & Distill        & \textbf{55.5} & 0.76 & 49.8 & 0.77 & \textbf{60.3} & 0.74 & 25.8 & 0.80 & 6.50 & 0.93 & 22.3 & 0.91 \\
      & DITTO        & 14.0 & 0.50 & 22.3 & 0.54 & 9.50 & 0.48 & 0.25 & 0.34 & 5.25 & 0.72 & 11.3 & 0.56  \\
      & DPO        & 0.00 & 0.88 & 0.00 & 0.87 & 1.25 & 0.87 & 2.50 & 0.78 & 6.25 & 0.83 & 24.3 & 0.95  \\
      & \textbf{RLSpoofer}   & 53.5 & 0.70 & \textbf{54.5} & 0.75 & 57.3 & 0.72 & \textbf{54.5} & 0.77 & \textbf{50.3} & 0.68 & \textbf{30.3} & 0.89  \\
    \midrule

    \multirow{4}{*}{\makecell[l]{Llama3.2-3B\\-Instruct}}
      & Distill        & 53.8 & 0.77 & 45.5 & 0.76 & \textbf{56.3} & 0.75 & 26.0 & 0.77 & 8.75 & 0.93 & 23.3 & 0.89 \\
      & DITTO        & 19.3 & 0.54 & 24.3 & 0.56 & 14.0 & 0.51 & 1.00 & 0.35 & 6.50 & 0.79 & 18.0 & 0.65  \\
      & DPO        & 2.50 & 0.49 & 0.50 & 0.53 & 0.75 & 0.36 & 7.75 & 0.60 & 6.25 & 0.67 & 25.0 & 0.87  \\
      & \textbf{RLSpoofer}   & \textbf{54.5} & 0.70 & \textbf{54.5} & 0.74 & 55.3 & 0.76 & \textbf{52.0} & 0.72 & \textbf{49.8} & 0.85 & \textbf{33.3} & 0.92  \\
    \midrule
    \bottomrule
  \end{tabular}
  \vspace{-0.3cm}
\end{table}

\begin{wraptable}{r}{0.34\textwidth}
  \centering
  \vspace{-0.4cm}
  \small  
  \setlength{\tabcolsep}{2.3pt}  
  \caption{Spoofing performance on EWD using Qwen3-4B.}
  \label{tab:logit_base_per}
  \begin{tabular}{@{}lcccc@{}}
    \toprule
    \textbf{Method}        & \textbf{SSR}   & \textbf{SR}   & \textbf{P-SP}   & \textbf{GPTS} \\
    \midrule
    Distill        & 51.3  & 62.5  & 0.81  & 7.66 \\
    DITTO     & 56.0  & \textbf{94.3}  & 0.68  & 6.66 \\
    DPO     & 0.25  & 0.50  & 0.78  & 7.55 \\
    RLSpoofer  & \textbf{56.5}  & 87.0  & 0.73  & 6.68 \\
    \bottomrule
  \end{tabular} 
\end{wraptable}
\textbf{Logit-based watermarks are vulnerable to baseline attacks.} 
We observe that across all attack settings, \textit{Distill} consistently achieves strong spoofing performance against logit-based watermarks. As shown in Table~\ref{tab:spoofmain}, Distill enables Qwen3-0.6B to achieve 42.3\% SSR and a P-SP score of 0.75 on EWD, significantly outperforming DITTO (7.5\%) and DPO (0.25\%), with similar trends observed across models. Furthermore, Table~\ref{tab:logit_base_per} shows that Distill and DITTO achieve spoof rates (SR) of 62.5\% and 94.3\%, respectively, on Qwen3-4B with EWD. This effectiveness arises because logit-based watermarking induces pronounced shifts in the output distribution~\cite{liu2024can}, which can be captured by distribution-matching approaches.
In contrast, DPO yields consistently low SSR, spoof rates, and P-SP scores across most watermarks, suggesting that preference-based optimization is less suitable for watermark spoofing.

\begin{wraptable}{r}{0.34\textwidth}
  \centering
  \vspace{-0.4cm}
  \small  
  \setlength{\tabcolsep}{2.3pt}  
  \caption{Spoofing performance on PF using Qwen3-4B.}
  \label{tab:samplePF_per}
  \begin{tabular}{@{}lcccc@{}}
    \toprule
    \textbf{Method}        & \textbf{SSR}   & \textbf{SR}   & \textbf{P-SP}   & \textbf{GPTS} \\
    \midrule
    Distill        & 6.00  & 6.50  & 0.96  & 7.67 \\
    DITTO     & 3.50  & 6.25  & 0.87  & 6.01 \\
    DPO     & 5.25  & 8.75  & 0.88  & 9.30 \\
    RLSpoofer  & \textbf{62.0}  & \textbf{82.8}  & 0.77  & 6.60 \\
    \bottomrule
  \end{tabular} 
\end{wraptable}
\textbf{Sampling-based distortion-free watermarks exhibit stronger resilience.}
We observe that all three baselines struggle to spoof sampling-based distortion-free watermarks effectively. Distill achieves only $\sim$7\% SSR on PF-Watermark and at most 23.3\% on PMark, with the other baselines showing similarly limited performance. Furthermore, as shown in Table~\ref{tab:samplePF_per}, all three methods obtain spoof rates below 9\% on PF-Watermark, indicating that they fail to replicate the watermark distribution. We hypothesize this resilience stems from the fact that PF and PMark watermarks are distortion-free in expectation~\cite{zhao2024permutePF}, rendering their latent signals exceptionally difficult to model from large training corpora.

\begin{figure}[t]
    \centering
    \includegraphics[width=\linewidth]{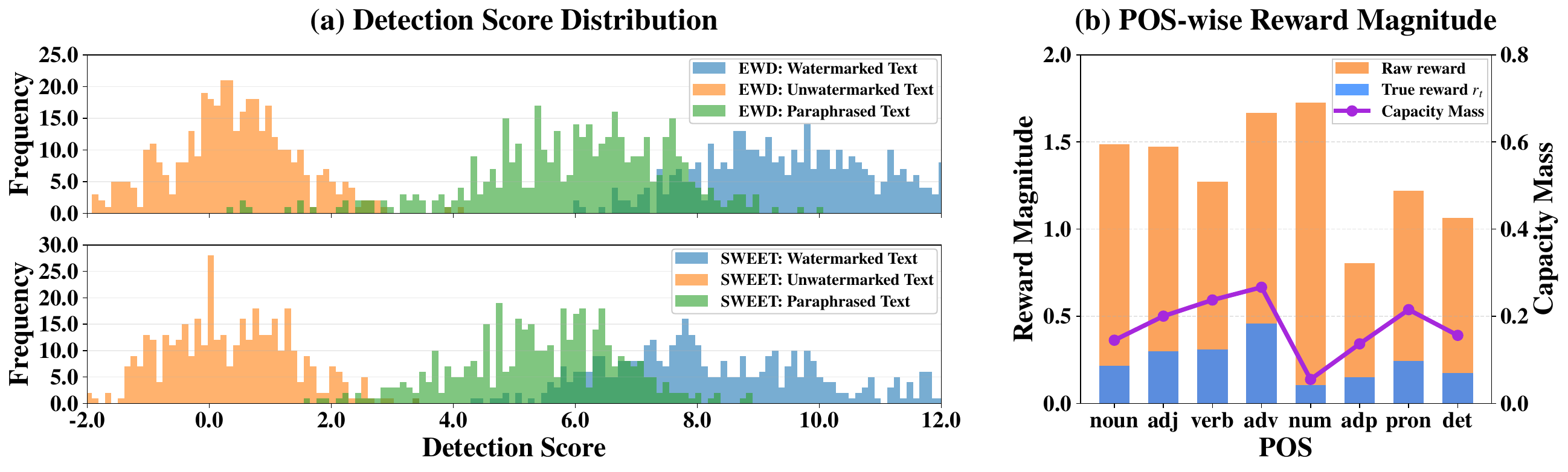}
    \caption{(a) illustrates the detection score distributions for EWD and SWEET watermarks across unwatermarked outputs, watermarked outputs, and paraphrased texts generated by Qwen3-4B trained with RLSpoofer; (b) shows the POS-wise reward magnitudes for SWEET watermark on Qwen3-4B.}
    \label{fig:detection_pos}
    \vspace{-0.3cm}
\end{figure}

\textbf{RLSpoofer consistently compromises the resilience of watermarking schemes.} We observe that, with only 100 training samples, RLSpoofer matches and surpasses baselines requiring 10,000 samples across all evaluated watermarks. Specifically, on SWEET, it enables Qwen3-0.6B to achieve a 50.5\% spoofing success rate (SSR), significantly exceeding Distill (20.0\%) and DITTO (6.75\%). Furthermore, RLSpoofer empowers Qwen3-4B to effectively spoof the PF watermark with a 62.0\% SSR, far surpassing the best baseline of 8.75\%. This efficacy stems from its ability to align with the distribution of watermarked text, successfully exploiting the sample-level distribution shifts induced even by distortion-free watermarks under a fixed key~\cite{liu2024can}. Our analysis of EWD and SWEET detection scores (400 samples, 500 tokens per input) confirms this alignment: as shown in Figure~\ref{fig:detection_pos}(a), RLSpoofer shifts the $z$-score distribution of its generated rephrasings toward that of watermarked text, maintaining clear separation from the unwatermarked distribution. Details are in Appendix~\ref{app:detection_scores}.


\subsection{Empirical Analysis of RLSpoofer}

In this subsection, we analyze the effectiveness of RLSpoofer on three complementary perspectives: feasible token-level watermark injection, sequence-level semantic preservation, and optimization stability. Additional details can be found in Appendix~\ref{app:local_capacity_mass} - \ref{app:ceanchor}.

\textbf{Local capacity mass identifies feasible room for watermark injection.}
Theorem~\ref{thm:capacity_clean} suggests that the local capacity mass \(c_t\) captures the local bottleneck of semantics-consistent redistribution, and thus favors positions with greater semantics-preserving substitutability. Since $c_t$ directly modulates token-level rewards, its effect should be reflected in the POS-wise distribution of $r_t$. We therefore plot POS-wise token-level rewards of the Qwen3-4B RLSpoofer on the SWEET training set. Figure~\ref{fig:detection_pos} (b) shows that the induced reward is small on numerals, which are typically critical for semantic fidelity, but larger on adjectives, verbs, and adverbs, which allow greater lexical flexibility. Thus, local capacity mass suppresses overly sharp rewards on semantically rigid tokens while amplifying rewards on more substitutable ones. This suggests that RLSpoofer exploits genuinely flexible positions rather than uniformly favoring all watermark-preferred tokens.

\begin{wraptable}{r}{0.32\textwidth}
  \centering
\small
\vspace{-0.4cm}
\caption{SSR on reward weights. }
\setlength{\tabcolsep}{2pt}  
\label{tab:capacity_weight}
\begin{tabular}{llcc}
\toprule
\textbf{Model} & \textbf{Weight} & \textbf{EWD} & \textbf{SWEET} \\
\midrule
\multirow{3}{*}{\makecell{Qwen3\\(0.6B)}}
& {1 - $p_{\max}$} & \textbf{54.3} & \textbf{50.5} \\
& uniform & 42.8 & 39.0 \\
& $p_{\max}$ & 40.5 & 29.8 \\
\midrule
\multirow{3}{*}{\makecell{Qwen3\\(4B)}}
& {1 - $p_{\max}$} & \textbf{56.5} & \textbf{52.3} \\
& uniform & 35.5 & 28.8 \\
& $p_{\max}$ & 28.0 & 25.0 \\
\bottomrule
\vspace{-0.5cm}
\end{tabular}
\end{wraptable}
We next verify that this weighting is not only intuitively aligned with semantic flexibility but also critical to spoofing performance. Specifically, we replace 1 - \(p_{\max}\) with either a uniform weight of 1 or the reverse weight \(p_{\max}\). Table~\ref{tab:capacity_weight} shows that both replacements consistently degrade performance. On EWD, uniform weighting and \(p_{\max}\) reduce the SSR of Qwen3-4B from 56.5\% to 35.5\% and 28.0\%, respectively; similar trends hold on SWEET and for Qwen3-0.6B. These results show that spoofing does not benefit from uniformly enlarging token-level rewards. Instead, it requires identifying contexts with sufficient semantically permissible redistribution room, where probability mass can be shifted toward watermark-preferred continuations without compromising semantic fidelity.

\begin{table}[t]
\centering
\small

\setlength{\tabcolsep}{5pt}  

\begin{minipage}[t]{0.49\textwidth}
\centering
\caption{SSR across semantic rewards.}
\label{tab:semantic_ablation}
\begin{tabular}{llcccc}
\toprule
\textbf{Model} & \textbf{Scheme} & \textbf{Min.} & \textbf{Avg.} & \textbf{Hum.} & \textbf{W.M.} \\
\midrule
\multirow{3}{*}{\makecell{Qwen3\\(0.6B)}}
& EWD   & \textbf{54.3} & 44.8 & 42.3 & 48.5 \\
& SWEET & \textbf{50.5} & 48.8 & 43.0 & 46.8 \\
& PF & \textbf{33.3} & 26.3  & 24.5 & 25.5 \\
\midrule
\multirow{3}{*}{\makecell{Qwen3\\(4B)}}
& EWD   & \textbf{56.5} & 47.3   & 47.5 & 45.5 \\
& SWEET & \textbf{52.3} & 34.3  & 47.0 & 48.5 \\
& PF & \textbf{62.0} & 50.3  & 48.5 & 51.0 \\
\bottomrule
\end{tabular}
\end{minipage}
\hfill
\begin{minipage}[t]{0.5\textwidth}
\centering
\caption{CE-anchor ablation (SSR vs. RLS.).}
\label{tab:ce_ablation}
\begin{tabular}{llcc}
\toprule
\textbf{Model} & \textbf{Scheme} & \textbf{W.O.} & \textbf{Dis.(100)} \\
\midrule
\multirow{3}{*}{\makecell{Qwen3\\(0.6B)}}
& EWD   & 29.3 {(-25.0)} & 0.25 {(-54.0)} \\
& SWEET & 27.3 {(-23.2)} & 0.25 {(-50.3)} \\
& PF & 12.3 {(-19.0)} & 0.00 {(-31.3)} \\
\midrule
\multirow{3}{*}{\makecell{Qwen3\\(4B)}}
& EWD   & 33.5 {(-23.0)} & 0.00 {(-56.5)} \\
& SWEET & 26.8 {(-25.5)} & 0.25 {(-52.0)} \\
& PF & 25.5 {(-36.5)} & 0.00 {(-62.0)} \\
\bottomrule
\end{tabular}
\end{minipage}
\end{table}

\textbf{A conservative semantic reward improves surrogate quality.}
Beyond identifying feasible token positions, successful spoofing must preserve the meaning of the original text while staying close to the watermarked rewrite, which serves as a surrogate for the watermark-favored distribution. We therefore use a conservative semantic reward ({Min.}), defined by the weaker of the two semantic matches. Compared with averaging the two matches ({Avg.}) or using only one reference ({Hum.}/{W.M.}), our design avoids one-sided alignment and produces a better surrogate target for optimization. As shown in Table~\ref{tab:semantic_ablation}, we observe Min. consistently achieves the best SSR across models and watermarking schemes, suggesting that successful spoofing benefits from preserving the original meaning while maintaining a better surrogate target for optimization.

\textbf{Cross-entropy (CE) anchoring stabilizes optimization.}
We observe that even with carefully designed token-level and semantic rewards, unconstrained policy optimization can still drift away from the base model and exploit brittle reward shortcuts. Cross-entropy anchoring mitigates this by keeping optimization aligned with the base-model distribution while still allowing watermark-favored redistribution. As shown in Table~\ref{tab:ce_ablation}, removing the anchor ({W.O.}) causes large SSR drops of 23.0 - 36.5 points, while replacing RL with supervised distillation on the same 100 training pairs ({Dis.(100)}) leads to near-complete failure. These results indicate that the gain comes from stable distributional alignment rather than limited-pair imitation alone.

\subsection{Ablation Study}

In this subsection, we study the sensitivity of RLSpoofer to the training data, the choice of surrogate distribution generated model, and its ability to generalize to OOD test data.


\begin{figure}[t]
    \centering

    \begin{minipage}[t]{0.35\textwidth}
        \vspace{0pt}
        \centering
        \small
        \setlength{\tabcolsep}{4pt}
        \renewcommand{\arraystretch}{1.03}

        \captionof{table}{SSR across Training set.}
        \label{tab:trainsamples}

        \begin{tabular}{@{}ll ccc@{}}
            \toprule
            & & \multicolumn{3}{c}{\textbf{Samples}} \\
            \cmidrule(lr){3-5}
            \textbf{Models} & \textbf{Scheme} & \textbf{50} & \textbf{100} & \textbf{200} \\
            \midrule
            \multirow{3}{*}{\makecell{Qwen3\\(0.6B)}} 
                & EWD   & 44.3 & 54.3 & 56.8 \\
                & SWEET & 42.5 & 50.5 & 51.3 \\
                & PF    & 12.3 & 31.3 & 20.5 \\
            \midrule
            \multirow{3}{*}{\makecell{Qwen3\\(4B)}} 
                & EWD   & 46.5 & 56.5 & 58.5 \\
                & SWEET & 43.0 & 52.3 & 54.3 \\
                & PF    & 20.8 & 60.8 & 25.3 \\
            \bottomrule
        \end{tabular}
    \end{minipage}
    \hfill
    \begin{minipage}[t]{0.62\textwidth}
        \vspace{0pt}
        \centering
        \includegraphics[width=0.93\linewidth]{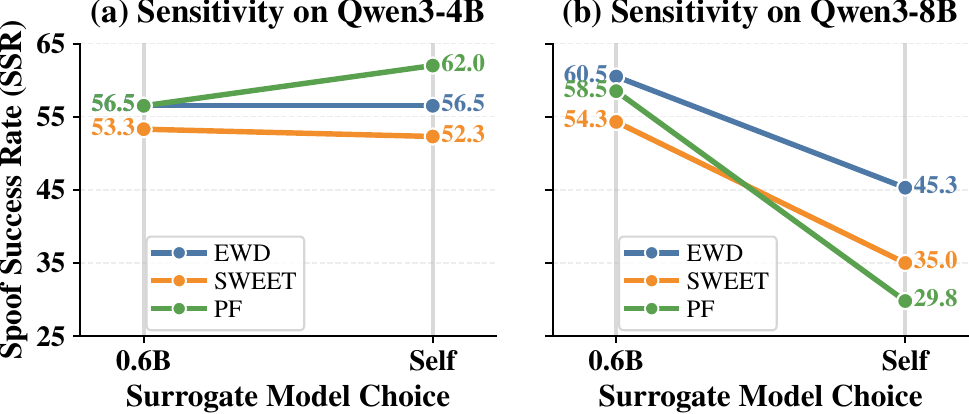}
        \caption{(a) and (b) compare SSR of RLSpoofer under two surrogate choices, Qwen3-0.6B (0.6B) and attack model (Self).}
        \label{fig:proxy_cross}
        \vspace{-0.3cm}
    \end{minipage}
\end{figure}

\textbf{Small training sets suffice for strong RLSpoofer performance.}
We examine the impact of training set size on RLSpoofer under three watermarking schemes, EWD, SWEET, and PF, using Qwen3-0.6B and Qwen3-4B. Table~\ref{tab:trainsamples} shows that RLSpoofer can achieve strong spoofing performance with only 50 training samples. For instance, on EWD, Qwen3-4B attains an SSR of 46.5\%, which increases to 58.5\% with 200 samples. A similar pattern holds for SWEET. However, for the PF watermark, enlarging the training set does not yield consistent gains in SSR. This is likely because PF's distortion-free design creates a distribution gap that becomes harder to capture with more data. Overall, these findings confirm that relatively small training sets are sufficient for RLSpoofer to obtain strong spoofing performance. We provide details in Appendix~\ref{app:training_set_size}.

\textbf{RLSpoofer exhibits increased sensitivity to surrogate selection as attacker capacity grows.} We evaluate how the surrogate model used to approximate the target watermarked distribution impacts RLSpoofer's efficacy. Specifically, we evaluate two surrogates, Qwen3-0.6B (0.6B) and the base attacker (Self), against EWD, SWEET, and PF watermarks across the 4B and 8B Qwen3 attackers.
Fig.~\ref{fig:proxy_cross} demonstrates that surrogate selection minimally affects the weaker Qwen3-4B, which yields comparable SSR across all settings. Conversely, the stronger Qwen3-8B is highly sensitive: the 0.6B surrogate consistently outperforms the Self surrogate. For instance, on the PF watermark, utilizing 0.6B instead of Self elevates the SSR from 29.8\% to 58.5\%, with substantial gains similarly observed for EWD and SWEET. We attribute this disparity to Qwen3-8B's superior watermark-removal capabilities~\cite{huang2025rlcracker}; its rewrites heavily weaken the embedded signal, rendering the model itself an ineffective proxy for the target watermark distribution. Details are in Appendix~\ref{app:sensitivity}.

\begin{wrapfigure}{r}{0.3\linewidth}
    \centering
    \vspace{-0.4cm}
    \includegraphics[width=\linewidth]{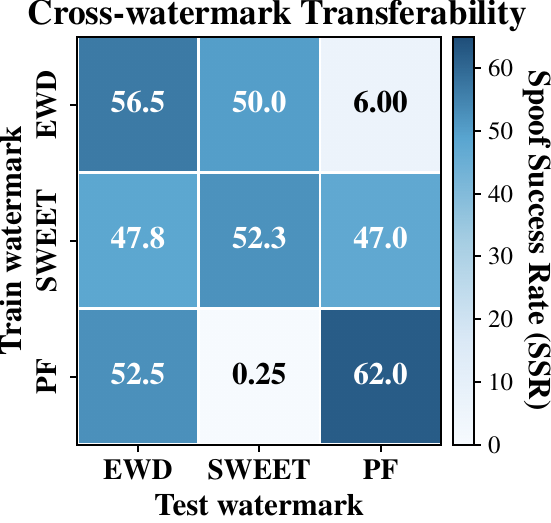}
    \caption{Cross watermark transferability on Qwen3-4B.}
    \label{fig:crosswm}
    \vspace{-0.3cm}
\end{wrapfigure}
\textbf{RLSpoofer demonstrates strongly asymmetric and directional cross-watermark transfer.}
We evaluate zero-shot transferability with Qwen3-4B by training RLSpoofer on one watermarking scheme and testing it on the others. As shown in Fig.~\ref{fig:crosswm}, transferability is not strictly constrained by watermark family. Although the KGW-style, logit-based watermarks EWD and SWEET~\cite{lu2024entropyEWD,lee2023wroteSWEET} exhibit substantial bidirectional transfer (SSR $\sim$ 47.8\% - 50\%), interactions involving the sampling-based PF watermark are notably directional. Specifically, RLSpoofer transfers effectively from SWEET to PF (SSR 47\%) and from PF to EWD (52.5\%), whereas the reverse directions are far weaker: PF to SWEET achieves only 0.25\% SSR, and EWD to PF only 6\%. This pronounced asymmetry suggests that transferability is not determined solely by mechanism-level similarity, but also by more intricate and directional overlaps in the vulnerabilities induced by different watermarking schemes. We provide details in Appendix~\ref{app:transferability}.

\section{Related Works}
\textbf{LLM watermarks.}
As LLMs become increasingly widespread, the human-written and machine-generated text grows harder to distinguish. By embedding imperceptible yet algorithmically detectable signals into generated text, watermarking has become important tool for content attribution, copyright protection, and the mitigation of malicious misuse~\citep{wu2025survey,liu2024preventing,zhang2024remark}. Existing methods mainly follow two paradigms. \textit{Logit-based} approaches~\citep{kirchenbauer2023watermarkKGW,lu2024entropyEWD} use a secret pseudo-random hash key to partition the vocabulary into \textit{green} and \textit{red} token lists at each generation step. They then bias the logits of green-list tokens to increase their sampling probability, and detection is performed through statistical tests for green-token overrepresentation~\citep{jovanovic2024watermark}. Conversely, \textit{sampling-based} methods~\citep{zhao2024permutePF,huo2025pmark, hou2024k} embed watermark signals directly into the sampling process, rather than modifying logits explicitly, offering finer-grained control while reducing quality degradation. Across both paradigms, recent advances prioritize minimizing statistical distortion, ensuring watermarked outputs remain indistinguishable from natural text~\citep{kuditipudi2023robust}.

\textbf{LLM watermark spoofing attack.}
LLM watermark \textit{spoofing attacks} seek to forge a target watermark signal and thereby falsely attribute arbitrary text to a specific model provider. Existing spoofing methods can generally be grouped into three categories. The first, \textit{piggyback spoofing}~\cite{pang2024attacking}, makes subtle edits to authentically watermarked text in order to inject toxic content while preserving detector confidence. However, this strategy substantially limits the attacker’s semantic flexibility. 
The second, \textit{feedback-guided spoofing} attacks~\cite{pang2024attacking,zhou2024bileve}, rely on repeated interaction with the watermark detector during generation to search for text that remains watermarked. Such attacks are typically query-intensive and depend on detector access, limiting their practical applicability. 
The third, \textit{learning-based spoofing}~\cite{jovanovic2024watermark,an2025ditto,gu2023learnabilitydistilled}, queries the target model to construct a training set for a surrogate model that internalizes the watermark signal. However, these approaches typically require either the knowledge of the watermarking mechanism~\cite{jovanovic2024watermark} or a large amount of training data (e.g., up to 10K samples~\cite{an2025ditto}). Consequently, the restrictive assumptions and inefficiencies of these approaches make them unsuitable for practically evaluating the spoofing resilience of watermarking schemes.

\section{Conclusion} 
In this paper, we study LLM watermark spoofing attack from a distributional perspective. We establish a local capacity bottleneck to characterize how token-level probability mass can be redistributed under KL-bounded, semantics-consistent local updates. Building on this insight, we propose RLSpoofer, a black-box RL-based spoofing attack that requires only 100 human–watermarked paraphrase pairs and no access to watermark internals or detectors. Across the evaluated setting of five attacker models and six watermarking schemes, RLSpoofer achieves strong spoofing performance on both logit-based and sampling-based watermarks. These results suggest that current LLM watermarking designs may remain vulnerable to spoofing attacks. We hope RLSpoofer can serve as a useful stress test for future watermark evaluation and for the development of more spoofing-resistant schemes.

\newpage
\bibliographystyle{unsrtnat}
\bibliography{reference}

\newpage
\appendix

\section{Details for the distributional surrogate}

\subsection{Absolute continuity and finite KL divergences}
\label{app:assumption}


In this subsection, we state the regularity conditions under which the distributional surrogate in Section~\ref{sec:preliminaries} is well defined and can be expanded into the expected log-likelihood-ratio objective in Eq.~\ref{eq:J_clean}. These conditions ensure that the relevant KL divergences are finite and that the log-likelihood ratio is integrable under the attack distribution.

\begin{assumption}[Absolute continuity and finite KL divergences]
\label{ass:regularity_clean}
For a fixed input \(\mathbf{X}\), \(P_\theta(\cdot \mid \mathbf{X})\) is absolutely continuous with respect to both reference distributions (\(P_\theta \ll P_h\) and \(P_\theta \ll P_{wm}\)). Furthermore, the log-likelihood ratio is integrable: \(\mathbb{E}_{\mathbf{X}' \sim P_\theta} [| \log(P_{wm}(\mathbf{X}' \mid \mathbf{X}) / P_h(\mathbf{X}' \mid \mathbf{X})) |] < \infty\). Moreover, $D_{\mathrm{KL}}\!\big(P_\theta(\cdot \mid \mathbf{X}) \,\|\, P_h(\cdot \mid \mathbf{X})\big) < \infty$ and $D_{\mathrm{KL}}\!\big(P_\theta(\cdot \mid \mathbf{X}) \,\|\, P_{wm}(\cdot \mid \mathbf{X})\big) < \infty.$
\end{assumption}

Under Assumption~\ref{ass:regularity_clean}, both KL terms in the surrogate objective are well defined, and their difference can be expanded as
\[
D_{\mathrm{KL}}\!\bigl(P_\theta \,\|\, P_h\bigr)
-
D_{\mathrm{KL}}\!\bigl(P_\theta \,\|\, P_{wm}\bigr)
=
\mathbb{E}_{\mathbf{X}' \sim P_\theta(\cdot \mid \mathbf{X})}
\left[
\log \frac{P_{wm}(\mathbf{X}' \mid \mathbf{X})}{P_h(\mathbf{X}' \mid \mathbf{X})}
\right],
\]
where we suppress the conditioning on \(\mathbf{X}\) for readability.

\subsection{Proof of Theorem~\ref{thm:capacity_clean}}\label{app:proof_capacity_clean}
In this subsection, we use the binary KL divergence
\[
d_{\mathrm{kl}}(p\|q)
:=
p\log\frac{p}{q}
+
(1-p)\log\frac{1-p}{1-q},
\qquad p,q\in[0,1],
\]
with the standard conventions $0\log(0/q):=0$ for $q\in[0,1]$, $0\log(0/0):=0$, and $b\log(b/0):=+\infty$ for $b>0$.

For convenience, we restate Theorem~\ref{thm:capacity_clean}.

\begin{theorem}[Local capacity characterization]
\label{thm:capacity_clean_restate}
Fix a history $h_t$. For any next-token distribution $\pi_\theta(\cdot\mid h_t)$ and any subset $A\subseteq\mathcal V$, we have
\begin{equation}
\resizebox{0.93\linewidth}{!}{$
\pi_\theta(A\mid h_t)-P_h(A\mid h_t)
\le
\sup\Bigl\{
\lambda\in[0,1]:
d_{\mathrm{kl}}\!\Bigl(\lambda \,\Big\|\, 1-\max_{x\in\mathcal V}P_h(x\mid h_t)\Bigr)
\le
D_{\mathrm{KL}}\!\bigl(\pi_\theta(\cdot\mid h_t)\,\|\,P_h(\cdot\mid h_t)\bigr)
\Bigr\}.
$}
\label{eq:capacity_bound_clean_restate}
\end{equation}
Moreover, for any fixed $C\ge 0$,
\begin{equation}
\lim_{1-\max_{x\in\mathcal V}P_h(x\mid h_t)\to 0}
\sup_{\pi_\theta:\,D_{\mathrm{KL}}(\pi_\theta(\cdot\mid h_t)\,\|\,P_h(\cdot\mid h_t))\le C}
\sup_{A\subseteq\mathcal V}
\bigl(\pi_\theta(A\mid h_t)-P_h(A\mid h_t)\bigr)=0.
\label{eq:capacity_vanish_clean_restate}
\end{equation}
\end{theorem}

\begin{proof}
Fix a history $h_t$, and let
$x_t^*\in\arg\max_{x\in\mathcal V}P_h(x\mid h_t)$ and
$q:=1-P_h(x_t^*\mid h_t)=1-\max_{x\in\mathcal V}P_h(x\mid h_t)$.
Thus, $q$ is the total mass assigned by $P_h(\cdot\mid h_t)$ to all tokens other than the dominant human-like continuation $x_t^*$.

We first prove Eq.~\eqref{eq:capacity_bound_clean_restate}. Consider the measurable map
$T:\mathcal V\to\{0,1\}$ defined by $T(x)=\mathbf 1\{x\neq x_t^*\}$, which induces the binary partition $\{x_t^*\}$ and $\mathcal V\setminus\{x_t^*\}$. Under $P_h(\cdot\mid h_t)$, the pushforward distribution of $T$ is $(1-q,q)$, whereas under $\pi_\theta(\cdot\mid h_t)$ it is $\bigl(\pi_\theta(x_t^*\mid h_t),\,1-\pi_\theta(x_t^*\mid h_t)\bigr)$. By the data processing inequality for KL divergence under the map $T$,
\begin{equation}
d_{\mathrm{kl}}\!\bigl(1-\pi_\theta(x_t^*\mid h_t)\,\|\,q\bigr)
\le
D_{\mathrm{KL}}\!\bigl(\pi_\theta(\cdot\mid h_t)\,\|\,P_h(\cdot\mid h_t)\bigr).
\label{eq:dpi_capacity_clean_restate}
\end{equation}
Hence,
\begin{equation}
1-\pi_\theta(x_t^*\mid h_t)
\le
\sup\Bigl\{
\lambda\in[0,1]:
d_{\mathrm{kl}}(\lambda\|q)
\le
D_{\mathrm{KL}}\!\bigl(\pi_\theta(\cdot\mid h_t)\,\|\,P_h(\cdot\mid h_t)\bigr)
\Bigr\}.
\label{eq:off_mode_mass_bound_clean_restate}
\end{equation}
Moreover, since $d_{\mathrm{kl}}(q\|q)=0$, we also have
\begin{equation}
q
\le
\sup\Bigl\{
\lambda\in[0,1]:
d_{\mathrm{kl}}(\lambda\|q)
\le
D_{\mathrm{KL}}\!\bigl(\pi_\theta(\cdot\mid h_t)\,\|\,P_h(\cdot\mid h_t)\bigr)
\Bigr\}.
\label{eq:q_feasible_clean_restate}
\end{equation}

Now fix any subset $A\subseteq\mathcal V$. We distinguish two cases.

\paragraph{Case 1: $x_t^*\notin A$.}
Since $A\subseteq \mathcal V\setminus\{x_t^*\}$, we have $\pi_\theta(A\mid h_t)\le 1-\pi_\theta(x_t^*\mid h_t)$. Therefore,
\[
\pi_\theta(A\mid h_t)-P_h(A\mid h_t)
\le
\pi_\theta(A\mid h_t)
\le
1-\pi_\theta(x_t^*\mid h_t).
\]

\paragraph{Case 2: $x_t^*\in A$.}
Then $A^c\subseteq \mathcal V\setminus\{x_t^*\}$, so $P_h(A^c\mid h_t)\le q$. Using
$\pi_\theta(A\mid h_t)-P_h(A\mid h_t)=P_h(A^c\mid h_t)-\pi_\theta(A^c\mid h_t)$,
we obtain
\[
\pi_\theta(A\mid h_t)-P_h(A\mid h_t)
\le
P_h(A^c\mid h_t)
\le q.
\]

Combining the two cases yields
\begin{equation}
\pi_\theta(A\mid h_t)-P_h(A\mid h_t)
\le
\max\Bigl\{1-\pi_\theta(x_t^*\mid h_t),\,q\Bigr\}.
\label{eq:set_shift_by_capacity_clean_restate}
\end{equation}
Finally, Eqs.~\eqref{eq:off_mode_mass_bound_clean_restate}, \eqref{eq:q_feasible_clean_restate}, and \eqref{eq:set_shift_by_capacity_clean_restate} imply
\[
\pi_\theta(A\mid h_t)-P_h(A\mid h_t)
\le
\sup\Bigl\{
\lambda\in[0,1]:
d_{\mathrm{kl}}(\lambda\|q)
\le
D_{\mathrm{KL}}\!\bigl(\pi_\theta(\cdot\mid h_t)\,\|\,P_h(\cdot\mid h_t)\bigr)
\Bigr\},
\]
which is exactly Eq.~\eqref{eq:capacity_bound_clean_restate}.

We next prove Eq.~\eqref{eq:capacity_vanish_clean_restate}. Fix $C\ge 0$, and define
\[
\rho(q;C):=
\sup\Bigl\{
\lambda\in[0,1]:
d_{\mathrm{kl}}(\lambda\|q)\le C
\Bigr\}.
\]
By Eq.~\eqref{eq:capacity_bound_clean_restate}, it suffices to show that
\begin{equation}
\rho(q;C)\to 0
\qquad\text{as } q\to 0.
\label{eq:rho_zero_capacity_clean_restate}
\end{equation}

Suppose otherwise. Then there exist $\varepsilon>0$, a sequence $q_n\to 0$, and $\lambda_n\in[\varepsilon,1]$ such that $d_{\mathrm{kl}}(\lambda_n\|q_n)\le C$ for all $n$. Since
\[
d_{\mathrm{kl}}(\lambda_n\|q_n)
=
\lambda_n\log\frac{\lambda_n}{q_n}
+
(1-\lambda_n)\log\frac{1-\lambda_n}{1-q_n},
\]
and the second term is bounded below by $(1-\lambda_n)\log(1-\lambda_n)$, we obtain
\[
d_{\mathrm{kl}}(\lambda_n\|q_n)
\ge
\lambda_n\log\frac{1}{q_n}
+
\lambda_n\log\lambda_n
+
(1-\lambda_n)\log(1-\lambda_n).
\]
Using the elementary bound $x\log x\ge -1/e$ on $[0,1]$ and $\lambda_n\ge\varepsilon$, it follows that
\[
d_{\mathrm{kl}}(\lambda_n\|q_n)
\ge
\varepsilon\log\frac{1}{q_n}-\frac{2}{e}.
\]
The right-hand side tends to $+\infty$ as $q_n\to 0$, contradicting $d_{\mathrm{kl}}(\lambda_n\|q_n)\le C$. This proves Eq.~\eqref{eq:rho_zero_capacity_clean_restate}.

Finally, for every $\pi_\theta$ satisfying
$D_{\mathrm{KL}}\!\bigl(\pi_\theta(\cdot\mid h_t)\,\|\,P_h(\cdot\mid h_t)\bigr)\le C$,
Eq.~\eqref{eq:capacity_bound_clean_restate} yields
\[
\sup_{A\subseteq\mathcal V}
\bigl(\pi_\theta(A\mid h_t)-P_h(A\mid h_t)\bigr)
\le
\rho(q;C).
\]
Taking the supremum over all such $\pi_\theta$, and then letting
$q=1-\max_{x\in\mathcal V}P_h(x\mid h_t)\to 0$,
Eq.~\eqref{eq:rho_zero_capacity_clean_restate} gives Eq.~\eqref{eq:capacity_vanish_clean_restate}.
\end{proof}

\section{Experimental Setup and Configuration}
\label{app:experimental_setup}
\subsection{Watermark algorithm setting}
\label{appen::algorithm_settings}

In this subsection, we report the hyperparameter settings for the watermarking algorithms evaluated in Section~\ref{sec:experiments}. For consistency and reproducibility, we use the official PMark~\cite{huo2025pmark} codebase\footnote{\url{https://github.com/PMark-repo/PMark}} with its default settings for both watermark generation and evaluation. For all other watermarking methods, we use the MarkLLM toolkit~\citep{pan2024markllm}\footnote{\url{https://github.com/THU-BPM/MarkLLM}} to generate watermarked text. Specifically, for KGW, we follow~\cite{an2025ditto} and set $\delta=3$, while for EWD, SWEET, PF, and Unigram, we use the default configurations provided by MarkLLM. MarkLLM is widely adopted in the watermarking literature due to its robustness and ease of integration.

\begin{tcolorbox}[
    title=Hyperparameters for the PMark watermark,
    colframe=gray, 
    colback=gray!15,
    coltitle=gray,
    fonttitle=\bfseries\color{white},
    rounded corners,
    enhanced,
    left=15pt, right=6pt, top=6pt, bottom=6pt,
    boxrule=1pt,
    arc=6pt,
    width=\linewidth
]
"algorithm\_name": "PMark",\\
"num\_samples": 64,\\
"pivot": rand,\\
"median\_method": prior,\\
"msig": 2
\end{tcolorbox}

\begin{tcolorbox}[
    title=Hyperparameters for the KGW watermark,
    colframe=gray, 
    colback=gray!15,
    coltitle=gray,
    fonttitle=\bfseries\color{white},
    rounded corners,
    enhanced,
    left=15pt, right=6pt, top=6pt, bottom=6pt,
    boxrule=1pt,
    arc=6pt,
    width=\linewidth
]
"algorithm\_name": "KGW",\\
"gamma": 0.5,\\
"delta": 3.0,\\
"hash\_key": 15485863,\\
"prefix\_length": 1,\\
"z\_threshold": 4.0,\\
"f\_scheme": "time",\\
"window\_scheme": "left"
\end{tcolorbox}

\begin{tcolorbox}[
    title=Hyperparameters for the EWD watermark,
    colframe=gray, 
    colback=gray!15,
    coltitle=gray,
    fonttitle=\bfseries\color{white},
    rounded corners,
    enhanced,
    left=15pt, right=6pt, top=6pt, bottom=6pt,
    boxrule=1pt,
    arc=6pt,
    width=\linewidth
]
"algorithm\_name": "EWD",\\
"gamma": 0.5,\\
"delta": 2.0,\\
"hash\_key": 15485863,\\
"prefix\_length": 1,\\
"z\_threshold": 4.0
\end{tcolorbox}

\begin{tcolorbox}[
    title=Hyperparameters for the SWEET watermark,
    colframe=gray, 
    colback=gray!15,
    coltitle=gray,
    fonttitle=\bfseries\color{white},
    rounded corners,
    enhanced,
    left=15pt, right=6pt, top=6pt, bottom=6pt,
    boxrule=1pt,
    arc=6pt,
    width=\linewidth
]
"algorithm\_name": "SWEET",\\
    "gamma": 0.5,\\
    "delta": 2.0,\\
    "hash\_key": 15485863,\\
    "z\_threshold": 4.0,\\
    "prefix\_length": 1,\\
    "entropy\_threshold": 0.9
\end{tcolorbox}

\begin{tcolorbox}[
    title=Hyperparameters for the PF watermark,
    colframe=gray, 
    colback=gray!15,
    coltitle=gray,
    fonttitle=\bfseries\color{white},
    rounded corners,
    enhanced,
    left=15pt, right=6pt, top=6pt, bottom=6pt,
    boxrule=1pt,
    arc=6pt,
    width=\linewidth
]
    "algorithm\_name": "PF",\\
    "ngram": 8,\\
    "seed": 0,\\
    "seeding": "hash",\\
    "salt\_key": 35317,\\
    "payload": 0,\\
    "max\_seq\_len": 8192
\end{tcolorbox}

\begin{tcolorbox}[
    title=Hyperparameters for the Unigram watermark,
    colframe=gray, 
    colback=gray!15,
    coltitle=gray,
    fonttitle=\bfseries\color{white},
    rounded corners,
    enhanced,
    left=15pt, right=6pt, top=6pt, bottom=6pt,
    boxrule=1pt,
    arc=6pt,
    width=\linewidth
]
    "algorithm\_name": "Unigram",\\
    "gamma": 0.5,\\
    "delta": 2.0,\\
    "hash\_key": 15485863,\\
    "z\_threshold": 4.0
\end{tcolorbox}

\subsection{Dataset Construction.}
\label{app:datatemplate}

We construct the training set from the C4-RealNewslike subset~\cite{raffel2020exploring} and the test set from Reddit WritingPrompts~\citep{verma2024ghostbuster}, LFQA~\citep{krishna2023paraphrasing}, and the BookReport and FakeNews subsets of MMW~\citep{piet2025markmywords}, following~\cite{jovanovic2024watermark}. To build training data, we first prompt Qwen3-8B to generate 500-token human-like unwatermarked texts. We then prompt watermarked Llama3.1-8B-Instruct to rewrite each text while preserving its semantics, yielding paired data of the form (human-like text, watermarked rewrite). For RLSpoofer, we use 100 such pairs for training and 20 for validation on each watermark. 

For DPO, we additionally obtain a rejected response by prompting Qwen3-8B to produce a non-watermarked rewrite of the same human-written text. The resulting preference pair consists of the watermarked rewrite as the chosen response and the non-watermarked rewrite as the rejected response. For distillation, we cast the task as supervised rewriting: the human-like text is formatted with the prompt template shown below, following~\cite{huang2025rlcracker}, and the corresponding watermarked rewrite is used as the target output for SFT.

\begin{tcolorbox}[
    title=Prompt template used,
    colframe=gray, 
    colback=gray!15,
    coltitle=gray,
    fonttitle=\bfseries\color{white},
    rounded corners,
    enhanced,
    left=6pt, right=6pt, top=6pt, bottom=6pt,
    boxrule=1pt,
    arc=6pt,
    width=\linewidth
]
\detokenize{
###Target Text: 
  }{ [human-written text]} 
  
\detokenize{
###Instruction: Rewrite the target text above using different words but keeping the same meaning and similar length. 
  }
  
\detokenize{
###Your Response:
  }
\end{tcolorbox}

For evaluation, we sample 400 unwatermarked test instances evenly from the four benchmark datasets. All training and test samples are standardized to 500 tokens.

\subsection{Baselines.}

\textbf{Spoof Methods.} Under the black-box threat model, in addition to RLSpoofer, we consider three spoofing baselines to further expose vulnerabilities in watermarking schemes. 
\textbf{Distill}~\cite{gu2023learnabilitydistilled} learns to imitate the rewriting distribution of the watermarked model via supervised fine-tuning on human--watermarked rewritten pairs. 
\textbf{DITTO}~\cite{an2025ditto} further exploits the distilled model by analyzing its output distribution and reproducing the statistical preferences induced by the target watermark. 
\textbf{DPO}~\cite{diaa2024optimizing} formulates watermark spoofing as a preference alignment problem, where the attacker is optimized to prefer watermarked rewrites over non-watermarked ones.

\textbf{Distill.} We implement Distill using LLaMA-Factory\footnote{\url{https://github.com/hiyouga/LLaMA-Factory}}~\cite{zheng2024llamafactory} and fine-tune each attacker model on 10{,}000 human--watermarked rewritten pairs. The task is formulated as a supervised rewriting problem: given an input human-written text wrapped with the prompt template shown above, the model is trained to generate the corresponding watermarked rewrite. We perform full-parameter fine-tuning for each model, using a learning rate of $2\times10^{-5}$ and a batch size of 128.

\textbf{DITTO.} We implement DITTO using its official codebase~\footnote{\url{https://github.com/hsannn/ditto}} and follow its default settings. In particular, DITTO is applied on top of the model obtained from the Distill stage, rather than being trained from scratch. That is, it starts from the distilled spoofing model and further enhances spoofing performance by matching the statistical preferences induced by the target watermark.

\textbf{DPO.} We also implement DPO using LLaMA-Factory. Following prior work, we construct 7{,}000 preference pairs, where the chosen response is the watermarked rewrite and the rejected response is a non-watermarked rewrite of the same human-written text. We perform full-parameter fine-tuning for each model with $\beta=0.1$, a batch size of 64, and a learning rate of $5\times10^{-6}$. We additionally tune the learning rate over $\{2\times10^{-5},\, 2\times10^{-6},\, 2\times10^{-7}\}$ and find that $5\times10^{-6}$ yields the best overall performance.

\textbf{Evaluation Details.} For evaluation, all spoofing methods use deterministic decoding with temperature $=0$. For Distill and DPO, we conduct inference with the vLLM engine to accelerate generation. For DITTO, we use the inference pipeline provided in its official implementation rather than vLLM. In all cases, the spoofing model takes the original input text and produces a rewritten response of similar semantics and length, which is then evaluated by the corresponding watermark detector. All reported spoofing results are obtained under this deterministic decoding setting.

\subsection{Implementation Details of RLSpoofer.}

RLSpoofer is implemented on top of the GRPO~\citep{shao2024deepseekmath} module in the TRL~\cite{vonwerra2022trl} library. The training data consists only of human--watermarked rewrite pairs \((\mathbf{X}, \mathbf{X}'_{wm})\), where \(\mathbf{X}\) is a human-written text and \(\mathbf{X}'_{wm}\) is its semantics-preserving rewrite produced by the target watermarked model. Given \(\mathbf{X}\), the attack policy \(\pi_\theta\) samples a group of \(G\) candidate rewrites \(\{x'_1,\dots,x'_G\}\), and is optimized by jointly combining a sequence-level semantic reward, a capacity-aware token-level reward, a token-wise KL regularizer, and a teacher-forced cross-entropy anchor. Next, we introduce the implementation details of each of the components in RLSpoofer.

\textbf{Reward Components.} We incorporate two reward signals to encourage semantic preservation and effective watermark spoofing:

\begin{itemize}[leftmargin=*, itemsep=1pt, parsep=0pt, topsep=0pt]
    \item \textbf{Semantic Reward.} The semantic reward \(A_i\in[-1,1]\) is computed based on the P-SP score~\citep{wieting2021paraphrastic} between the generated output and both the original watermarked response \(\mathbf{X}'_{wm}\) and the original human-written text \(\mathbf{X}\). Specifically, we first compute the semantic similarity between the generated output and each reference, and then take the \textit{minimum} of the two scores as the final semantic score. To enhance gradient flow and emphasize semantic fidelity, we apply a sigmoid-based scaling with a threshold of \(0.85\) following~\cite{huang2025rlcracker}. This maps semantic scores in the range \([0.7,1.0]\) to reward values between \(-1\) and \(1\). Specifically,
    \[
    A_i = \frac{2}{1 + e^{-x}} - 1,
    \qquad
    \text{where } x = \log\left(\frac{0.975}{0.025}\right)\cdot \frac{\text{P-SP score}-0.85}{1-0.85}.
    \]
    The advantage term is normalized as $\hat{A}_i = \frac{A_i - \mathrm{mean}(A)}{\mathrm{std}(A)+10^{-6}}.$

    \item \textbf{Capacity-aware Token-level Reward.} This reward encourages each sampled token to align with the watermark-induced surrogate distribution while staying aware of the locally available semantics-preserving redistribution room. In implementation, to instantiate the surrogate human-like and watermarked distributions, we query the same reference model \(\pi_{\mathrm{ref}}\) with the rewriting template in Appendix~\ref{app:datatemplate}, following~\cite{huang2025rlcracker}. Specifically, the human-written text \(\mathbf{X}\) and the watermarked rewrite \(\mathbf{X}'_{wm}\) are respectively filled into the same template as the target text, yielding two conditioning contexts that share the same rewriting instruction. For each sampled token \(x'_{i,t}\), we then compute its log-probability under the human-conditioned surrogate distribution, \(\log p_h(x'_{i,t}) = \log \pi_{\mathrm{ref}}(x'_{i,t}\mid \mathbf{X}, x'_{i,<t})\), and under the watermarked-conditioned surrogate distribution, \(\log p_{wm}(x'_{i,t}) = \log \pi_{\mathrm{ref}}(x'_{i,t}\mid \mathbf{X}'_{wm}, x'_{i,<t})\). We further define the local capacity mass as
    \[
    c_{i,t} = 1 - \max_{v\in\mathcal{V}} \pi_{\mathrm{ref}}(v\mid \mathbf{X}, x'_{i,<t}),
    \]
    which measures the probability mass outside the dominant human-like continuation. Based on this quantity, the capacity-aware token-level reward is defined as
    \[
    r_{i,t} = c_{i,t}\Bigl(\log p_{wm}(x'_{i,t}) - \log p_h(x'_{i,t})\Bigr).
    \]
    Here, \(\log p_{wm}(x'_{i,t}) - \log p_h(x'_{i,t})\) measures whether the sampled token is more preferred by the watermark-conditioned surrogate than by the human-conditioned surrogate, while \(c_{i,t}\) downweights positions where the human-like distribution is already sharply concentrated. As a result, the reward is amplified only at positions with sufficient local flexibility for semantics-preserving redistribution.
\end{itemize}

\textbf{Cross-entropy anchor.} To stabilize optimization, we further introduce a \textbf{cross-entropy anchor}. Specifically, we condition the policy on the human-written input \(\mathbf{X}\) and teacher-force it to predict the watermarked rewrite \(\mathbf{X}'_{wm}\), yielding $\mathcal{L}_{\mathrm{CE}}=-\sum_t \log \pi_\theta(x^{wm}_t\mid \mathbf{X},x^{wm}_{<t}).$ This term anchors the policy toward the target watermarked rewriting distribution and prevents unstable distributional drift during RL training.

Putting these components together, the implemented objective can be written as
\begin{align*}
    \mathcal{J}(\theta) \approx\;
    \mathbb{E}_{\{x'_i\}\sim\pi_{\theta_{\mathrm{old}}}}
    \frac{1}{G} \sum_{i=1}^G \frac{1}{|x'_i|} \sum_{t=1}^{|x'_i|} &
    \Biggl[
      \frac{\pi_\theta(x'_{i,t}\mid \mathbf{X}, x'_{i,<t})}
           {\pi_{\theta_{\mathrm{old}}}(x'_{i,t}\mid \mathbf{X}, x'_{i,<t})}
      \bigl(w_1 \hat{A}_i + w_2 r_{i,t}\bigr) \\
    &\qquad
      - \beta D_{\mathrm{KL}}\bigl[\pi_\theta \,\|\, \pi_{\mathrm{ref}}\bigr]
    \Biggr]
    - w_3\,\mathcal{L}_{\mathrm{CE}}\bigl(\pi_\theta, \{(\mathbf{X}, \mathbf{X}'_{wm})\}\bigr), 
\end{align*}
In implementation, the KL regularizer is directly adopted from TRL-GRPO, and we set \(\beta=0.04\) following the OpenR1~\cite{openr1} setting. Unless otherwise specified, the reference model \(\pi_{\mathrm{ref}}\) is chosen as the initial policy model before RL training.

\textbf{Hyperparameters.}
For each watermarking scheme, we train a dedicated RLSpoofer model. The training set consists of 100 human--watermarked rewrite pairs, each standardized to 500 tokens. In addition, we use a validation set of 20 data pairs to select the hyperparameters \(w_1\), \(w_2\), and \(w_3\). We train the model for 10 epochs with a batch size of 48 and a group size of \(G=12\), using a cosine learning rate scheduler. Empirically, we find that a relatively large learning rate is important for inducing a sufficient distribution shift. Specifically, we use a learning rate of \(2\times 10^{-5}\) for PMark and \(2\times 10^{-4}\) and $1\times 10^{-4}$ for the other watermarking schemes. Training takes approximately 1.5 hours for Qwen3-4B and 45 minutes for Qwen3-0.6B on four NVIDIA Pro 6000 Blackwell GPUs.

We find that RLSpoofer is sensitive to the relative weights of its three optimization components, \(w_1\), \(w_2\), and \(w_3\), and different choices of these weights lead to substantial variation in performance. To study this effect, we train Qwen3-4B to spoof the EWD watermark and analyze, in Figure~\ref{fig:W_Sensitivity}, how the spoof rate and P-SP score change as each weight varies. Here, the spoof rate refers to the proportion of rewritten outputs detected as watermarked, without imposing any semantic similarity constraint. In the sensitivity analysis, we vary one weight at a time while fixing the other two: for \(w_1\), we fix \(w_2=2\) and \(w_3=1\); for \(w_2\), we fix \(w_1=3\) and \(w_3=1\); and for \(w_3\), we fix \(w_1=3\) and \(w_2=2\). Overall, we find that \((w_1,w_2,w_3)=(3,2,1)\) provides the best trade-off between spoofing strength and semantic fidelity for EWD. For completeness and reproducibility, we report the final weight settings for each watermarking scheme and attack model in Table~\ref{tab:hyperparamsRLS}.

\begin{figure}[ht]
    \centering
    \includegraphics[width=\linewidth]{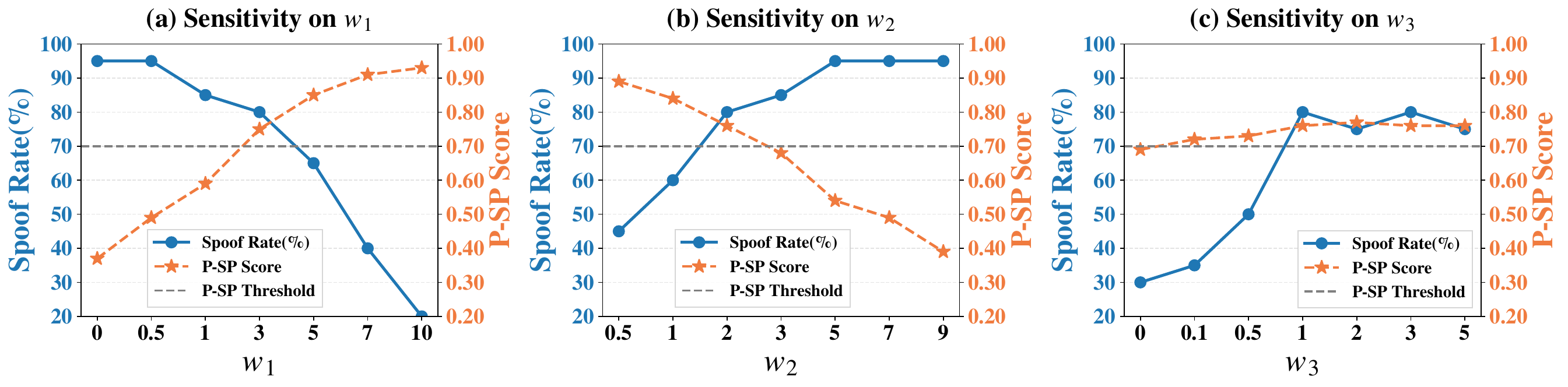}
    \caption{Sensitivity of RLSpoofer across component weights on validation set.}
    \label{fig:W_Sensitivity}
\end{figure}

\begin{table}[t]
\centering
\small
\setlength{\tabcolsep}{4pt}
\renewcommand{\arraystretch}{0.95}
\caption{Hyperparameter configurations of RLSpoofer across watermarking schemes.}
\label{tab:hyperparamsRLS}
\begin{tabular}{llcccccc}
\toprule
\textbf{Model} & \textbf{Hyper.} & \textbf{EWD} & \textbf{SWEET} & \textbf{KGW} & \textbf{Unigram} & \textbf{PF} & \textbf{PMark} \\
\midrule
\multirow{4}{*}{Qwen3-0.6B}
& $w_1$ & 5 & 3 & 5 & 3 & 1 & 1 \\
& $w_2$ & 10 & 2 & 1 & 5 & 5 & 2 \\
& $w_3$ & 1 & 1 & 1 & 1 & 2 & 1 \\
& lr    & $2\times10^{-4}$ & $2\times10^{-4}$ & $2\times10^{-4}$ & $2\times10^{-4}$ & $2\times10^{-4}$ & $2\times10^{-5}$ \\
\midrule
\multirow{4}{*}{Qwen3-1.7B}
& $w_1$ & 5 & 5 & 3 & 8 & 0.1 & 1 \\
& $w_2$ & 5 & 5 & 2 & 5 & 5 & 5 \\
& $w_3$ & 1 & 2 & 1 & 1 & 2 & 1 \\
& lr    & $2\times10^{-4}$ & $2\times10^{-4}$ & $2\times10^{-4}$ & $2\times10^{-4}$ & $2\times10^{-4}$ & $2\times10^{-5}$ \\
\midrule
\multirow{4}{*}{Qwen3-4B}
& $w_1$ & 3 & 3 & 5 & 8 & 1 & 1 \\
& $w_2$ & 2 & 1 & 1 & 5 & 10 & 5 \\
& $w_3$ & 1 & 1 & 1 & 1 & 3 & 1 \\
& lr    & $2\times10^{-4}$ & $2\times10^{-4}$ & $2\times10^{-4}$ & $2\times10^{-4}$ & $2\times10^{-4}$ & $2\times10^{-5}$ \\
\midrule
\multirow{4}{*}{Qwen2.5-3B-Instruct}
& $w_1$ & 3 & 3 & 5 & 5 & 3 & 3 \\
& $w_2$ & 2 & 2 & 2 & 2 & 5 & 10 \\
& $w_3$ & 1 & 1 & 2 & 1 & 2 & 2 \\
& lr    & $2\times10^{-4}$ & $2\times10^{-4}$ & $1\times10^{-4}$ & $2\times10^{-4}$ & $2\times10^{-4}$ & $2\times10^{-5}$ \\
\midrule
\multirow{4}{*}{Llama3.2-3B-Instruct}
& $w_1$ & 5 & 5 & 1 & 1 & 1 & 1 \\
& $w_2$ & 2 & 2 & 1 & 2 & 2 & 5 \\
& $w_3$ & 2 & 2 & 1 & 1 & 3 & 1 \\
& lr    & $2\times10^{-4}$ & $1\times10^{-4}$ & $1\times10^{-4}$ & $1\times10^{-4}$ & $1\times10^{-4}$ & $2\times10^{-5}$ \\
\bottomrule
\end{tabular}
\end{table}

\section{Experimental Results and Analysis}
\label{app:experimental_details}

\subsection{Robustness of watermarking schemes against spoofing}
\label{app:main_results}

We report the full spoofing results across six watermarking schemes, five attacker models, and four spoofing methods in Tables~\ref{tab:ewd_sweet_main}, \ref{tab:kgw_unigram_main}, and \ref{tab:pf_pmark_main}. We evaluate all methods using Spoof Success Rate (SSR), Spoof Rate (SR), P-SP, perplexity (PPL), and GPT-score (GPTS). We treat SSR as the primary metric, since it measures the proportion of rewritten outputs that are both detected as watermarked and satisfy the semantic constraint \( \text{P-SP}>0.7 \), whereas SR measures detector success without semantic filtering. We report the results averaged across 2 random seeds, 42 and 1234. For GPTS, we use GPT-5.4-nano as the judge model with the prompt template below, following~\cite{huang2025rlcracker}.

\begin{tcolorbox}[
    title=Prompt Template for GPT-5.4-nano as a Judge,
    colframe=gray,
    colback=gray!10,
    coltitle=black,
    fonttitle=\bfseries,
    rounded corners,
    enhanced,
    left=6pt, right=6pt, top=6pt, bottom=6pt,
    boxrule=0.8pt,
    arc=6pt,
    width=\linewidth
]
\detokenize{
    role: system,
    content: You are an impartial evaluator. You will be provided with an original text and a paraphrased version of that text. 
        Your task is to assess the quality of paraphrased text based on the following criteria: 
        1. The degree to which the paraphrased text maintains the meaning of the original text.
        2. The fluency, coherence, and clarity of the paraphrased text.
        3. The extent to which the style and tone of the paraphrased text match those of the original text.
        Please provide an objective evaluation and rate the paraphrased text on a scale of 1 to 10:
        - A rating of 1 indicates that the paraphrased text completely alters the meaning of the original text.
        - A rating of 10 indicates that the paraphrased text preserves the exact meaning of the original text, maintaining fluency, clarity, and style.
        Your response should return the score only.}
  
\detokenize{
    role: user,
    content: You will receive a original text and a pharased text. Please act as an impartial judge and evaluate the quality of the paraphrased text. You should evaluate the paraphrased text based on the following criteria:
    1. How closely the meaning of the paraphrased text aligns with the original text.
    2. The fluency and clarity of the paraphrased text.
    3. How closely the style and tone of the paraphrased text align with the original text.
    Rate the paraphrased text on a scale of 1 to 10, where:
    - A rating of 1 indicates that the paraphrased text deviates significantly from the original meaning.
    - A rating of 10 indicates that the paraphrased text perfectly preserves the original meaning, fluency, clarity, and tone.
  }
  
Your response should strictly follow this format and return the score only: [Rating]. 

Here  the original text: { [human-written text]} 

Here’s the paraphrased text:{ [rephrased text]}
\end{tcolorbox}

\newcommand{\stdinline}[2]{#1 {\scriptsize$\pm$ #2}}
\newcommand{\beststdinline}[2]{\textbf{#1} {\scriptsize$\pm$ #2}}

\begin{table}[ht]
  \centering
  \caption{Spoof Success Rate (SSR, \%), Spoof Rate (SR, \%), P-SP, Perplexity (PPL), and GPT-score (GPTS) across models. For readability, we report SSR and SR as mean $\pm$ standard deviation over two random seeds. \textbf{Bold} indicates the best mean SSR for each model under each watermarking scheme. Higher SSR indicates stronger spoofing performance.}
  \label{tab:ewd_sweet_main}
  \small
  \setlength{\tabcolsep}{2pt}
  \renewcommand{\arraystretch}{0.95}
  \begin{tabular}{@{}llcccccccccc@{}}
    \toprule
    & & \multicolumn{5}{c}{\textbf{EWD}} & \multicolumn{5}{c}{\textbf{SWEET}} \\
    \cmidrule(lr){3-7} \cmidrule(lr){8-12}
    \textbf{Models} & \textbf{Methods}
    & \textbf{SSR} & \textbf{SR} & \textbf{P-SP} & \textbf{PPL} & \textbf{GPTS}
    & \textbf{SSR} & \textbf{SR} & \textbf{P-SP} & \textbf{PPL} & \textbf{GPTS} \\
    \midrule

    \multirow{4}{*}{Qwen3-0.6B}
      & Distill
      & \stdinline{42.3}{1.2} & \stdinline{57.3}{2.5} & 0.75 & 4.50 & 6.62
      & \stdinline{20.0}{0.8} & \stdinline{33.8}{1.2} & 0.76 & 4.68 & 6.62 \\
      & DITTO
      & \stdinline{7.50}{0.5} & \stdinline{97.3}{0.5} & 0.43 & 4.78 & 3.00
      & \stdinline{6.75}{0.3} & \stdinline{95.8}{0.6} & 0.41 & 4.52 & 3.00 \\
      & DPO
      & \stdinline{0.25}{0.8} & \stdinline{0.75}{0.0} & 0.57 & 1.95 & 4.42
      & \stdinline{0.00}{0.0} & \stdinline{0.00}{0.0} & 0.76 & 2.71 & 4.42 \\
      & RLSpoofer
      & \beststdinline{54.3}{0.8} & \stdinline{80.5}{1.8} & 0.73 & 5.36 & 5.77
      & \beststdinline{50.5}{1.2} & \stdinline{66.3}{2.0} & 0.79 & 4.93 & 5.77 \\
    \midrule

    \multirow{4}{*}{Qwen3-1.7B}
      & Distill
      & \stdinline{43.8}{1.5} & \stdinline{53.8}{2.2} & 0.80 & 4.59 & 7.27
      & \stdinline{26.8}{2.5} & \stdinline{37.8}{1.4} & 0.79 & 4.63 & 7.27 \\
      & DITTO
      & \stdinline{21.5}{0.8} & \stdinline{96.0}{0.2} & 0.53 & 5.26 & 4.68
      & \stdinline{29.0}{0.7} & \stdinline{92.8}{0.8} & 0.56 & 5.22 & 4.68 \\
      & DPO
      & \stdinline{0.25}{0.0} & \stdinline{0.25}{0.0} & 0.94 & 2.35 & 8.17
      & \stdinline{0.00}{0.0} & \stdinline{0.00}{0.0} & 0.84 & 2.46 & 8.17 \\
      & RLSpoofer
      & \beststdinline{53.5}{1.2} & \stdinline{68.5}{0.7} & 0.76 & 5.15 & 6.50
      & \beststdinline{52.0}{0.8} & \stdinline{72.0}{0.6} & 0.71 & 5.92 & 6.45 \\
    \midrule

    \multirow{4}{*}{Qwen3-4B}
      & Distill
      & \stdinline{51.3}{0.6} & \stdinline{62.5}{1.5} & 0.81 & 4.75 & 7.66
      & \stdinline{37.3}{0.8} & \stdinline{43.8}{2.5} & 0.82 & 4.79 & 7.66 \\
      & DITTO
      & \stdinline{56.0}{0.3} & \stdinline{94.3}{0.2} & 0.68 & 5.62 & 6.66
      & \stdinline{43.3}{0.6} & \stdinline{94.5}{0.3} & 0.66 & 5.25 & 6.66 \\
      & DPO
      & \stdinline{0.25}{0.0} & \stdinline{0.50}{0.0} & 0.78 & 2.24 & 7.55
      & \stdinline{0.00}{0.0} & \stdinline{0.00}{0.0} & 0.78 & 2.02 & 7.55 \\
      & RLSpoofer
      & \beststdinline{56.5}{0.6} & \stdinline{87.0}{0.4} & 0.73 & 4.88 & 6.68
      & \beststdinline{52.3}{0.2} & \stdinline{71.3}{0.3} & 0.75 & 5.38 & 6.78 \\
    \midrule

    \multirow{4}{*}{\makecell[l]{Qwen2.5-3B\\-Instruct}}
      & Distill
      & \beststdinline{55.5}{0.5} & \stdinline{81.5}{1.2} & 0.76 & 4.40 & 7.31
      & \stdinline{49.8}{0.3} & \stdinline{67.5}{1.0} & 0.77 & 4.37 & 7.31 \\
      & DITTO
      & \stdinline{14.0}{0.0} & \stdinline{99.5}{0.0} & 0.50 & 4.95 & 5.32
      & \stdinline{22.3}{0.2} & \stdinline{99.8}{0.0} & 0.54 & 4.78 & 5.32 \\
      & DPO
      & \stdinline{0.00}{0.0} & \stdinline{0.75}{0.0} & 0.88 & 2.00 & 7.76
      & \stdinline{0.00}{0.0} & \stdinline{0.50}{0.0} & 0.87 & 1.94 & 7.76 \\
      & RLSpoofer
      & \stdinline{53.5}{0.6} & \stdinline{83.0}{0.8} & 0.70 & 4.66 & 6.80
      & \beststdinline{54.5}{1.2} & \stdinline{68.8}{1.3} & 0.75 & 6.64 & 6.57 \\
    \midrule

    \multirow{4}{*}{\makecell[l]{Llama3.2-3B\\-Instruct}}
      & Distill
      & \stdinline{53.8}{0.2} & \stdinline{71.3}{1.3} & 0.77 & 4.49 & 7.22
      & \stdinline{45.5}{0.5} & \stdinline{62.0}{2.5} & 0.76 & 4.37 & 7.22 \\
      & DITTO
      & \stdinline{19.3}{0.3} & \stdinline{99.5}{0.0} & 0.54 & 4.90 & 5.13
      & \stdinline{24.3}{0.2} & \stdinline{100.}{0.0} & 0.56 & 4.76 & 5.13 \\
      & DPO
      & \stdinline{2.50}{0.0} & \stdinline{1.25}{0.0} & 0.49 & 2.00 & 2.62
      & \stdinline{0.50}{0.0} & \stdinline{3.00}{0.0} & 0.53 & 1.86 & 2.62 \\
      & RLSpoofer
      & \beststdinline{54.5}{0.8} & \stdinline{79.5}{0.3} & 0.70 & 4.99 & 6.79
      & \beststdinline{54.5}{1.0} & \stdinline{79.5}{2.2} & 0.74 & 6.15 & 6.43 \\
    \bottomrule
  \end{tabular}
\end{table}

\begin{table}[t]
  \centering
  \caption{Spoof Success Rate (SSR, \%), Spoof Rate (SR, \%), P-SP, Perplexity (PPL), and GPT-score (GPTS) across models. For readability, we report SSR and SR as mean $\pm$ standard deviation over two random seeds. \textbf{Bold} indicates the best mean SSR for each model under each watermarking scheme. Higher SSR indicates stronger spoofing performance.}
  \label{tab:kgw_unigram_main}
  \small
  \setlength{\tabcolsep}{2pt}
  \renewcommand{\arraystretch}{0.95}
  \begin{tabular}{@{}llcccccccccc@{}}
    \toprule
    & & \multicolumn{5}{c}{\textbf{KGW}} & \multicolumn{5}{c}{\textbf{Unigram}} \\
    \cmidrule(lr){3-7} \cmidrule(lr){8-12}
    \textbf{Models} & \textbf{Methods}
    & \textbf{SSR} & \textbf{SR} & \textbf{P-SP} & \textbf{PPL} & \textbf{GPTS}
    & \textbf{SSR} & \textbf{SR} & \textbf{P-SP} & \textbf{PPL} & \textbf{GPTS} \\
    \midrule

    \multirow{4}{*}{Qwen3-0.6B}
      & Distill
      & \stdinline{35.8}{0.5} & \stdinline{66.8}{3.5} & 0.68 & 4.78 & 5.91
      & \stdinline{13.8}{0.2} & \stdinline{27.5}{1.5} & 0.84 & 2.62 & 5.91 \\
      & DITTO
      & \stdinline{1.00}{0.0} & \stdinline{84.0}{0.8} & 0.33 & 5.05 & 2.34
      & \stdinline{0.25}{0.0} & \stdinline{99.3}{0.0} & 0.32 & 1.67 & 2.34 \\
      & DPO
      & \stdinline{1.00}{0.0} & \stdinline{11.0}{0.6} & 0.66 & 2.12 & 5.31
      & \stdinline{0.25}{0.0} & \stdinline{13.8}{1.2} & 0.32 & 1.47 & 5.31 \\
      & RLSpoofer
      & \beststdinline{52.0}{0.6} & \stdinline{82.5}{0.8} & 0.72 & 5.16 & 5.32
      & \beststdinline{49.5}{0.8} & \stdinline{82.3}{1.0} & 0.70 & 5.16 & 5.43 \\
    \midrule

    \multirow{4}{*}{Qwen3-1.7B}
      & Distill
      & \stdinline{44.5}{0.5} & \stdinline{67.5}{2.5} & 0.75 & 5.11 & 7.24
      & \stdinline{19.5}{0.8} & \stdinline{38.0}{0.6} & 0.81 & 2.50 & 7.24 \\
      & DITTO
      & \stdinline{13.8}{0.2} & \stdinline{95.5}{0.0} & 0.52 & 6.36 & 4.43
      & \stdinline{1.50}{0.0} & \stdinline{100.}{0.0} & 0.36 & 1.76 & 4.43 \\
      & DPO
      & \stdinline{1.00}{0.0} & \stdinline{1.00}{0.0} & 0.96 & 2.55 & 9.12
      & \stdinline{0.50}{0.0} & \stdinline{2.50}{0.0} & 0.87 & 2.03 & 9.12 \\
      & RLSpoofer
      & \beststdinline{52.0}{1.2} & \stdinline{78.8}{2.2} & 0.71 & 5.62 & 6.24
      & \beststdinline{54.8}{0.8} & \stdinline{83.8}{1.3} & 0.73 & 4.72 & 6.53 \\
    \midrule

    \multirow{4}{*}{Qwen3-4B}
      & Distill
      & \stdinline{57.0}{0.5} & \stdinline{74.3}{1.5} & 0.79 & 5.23 & 7.67
      & \stdinline{28.0}{0.2} & \stdinline{48.3}{3.8} & 0.80 & 2.74 & 7.67 \\
      & DITTO
      & \stdinline{36.5}{0.8} & \stdinline{98.3}{0.0} & 0.61 & 6.64 & 6.01
      & \stdinline{2.25}{0.0} & \stdinline{100.}{0.0} & 0.39 & 2.01 & 6.01 \\
      & DPO
      & \stdinline{1.00}{0.0} & \stdinline{1.50}{0.0} & 0.93 & 2.43 & 9.30
      & \stdinline{0.75}{0.0} & \stdinline{13.8}{0.7} & 0.59 & 1.67 & 9.30 \\
      & RLSpoofer
      & \beststdinline{58.0}{0.7} & \stdinline{79.5}{2.5} & 0.75 & 5.80 & 7.25
      & \beststdinline{54.8}{0.5} & \stdinline{88.5}{1.7} & 0.74 & 4.12 & 6.66 \\
    \midrule

    \multirow{4}{*}{\makecell[l]{Qwen2.5-3B\\-Instruct}}
      & Distill
      & \beststdinline{60.3}{0.3} & \stdinline{88.8}{0.2} & 0.74 & 4.89 & 7.43
      & \stdinline{25.8}{1.0} & \stdinline{42.5}{2.5} & 0.80 & 2.60 & 7.43 \\
      & DITTO
      & \stdinline{9.50}{0.5} & \stdinline{99.8}{0.0} & 0.48 & 6.16 & 4.67
      & \stdinline{0.25}{0.3} & \stdinline{100.}{0.0} & 0.34 & 2.09 & 4.67 \\
      & DPO
      & \stdinline{1.25}{0.0} & \stdinline{6.25}{0.0} & 0.87 & 2.16 & 7.15
      & \stdinline{2.50}{0.0} & \stdinline{2.50}{0.0} & 0.78 & 1.87 & 7.15 \\
      & RLSpoofer
      & \stdinline{57.3}{0.3} & \stdinline{80.5}{0.7} & 0.72 & 5.81 & 6.73
      & \beststdinline{54.5}{0.5} & \stdinline{76.5}{1.3} & 0.77 & 4.55 & 6.64 \\
    \midrule

    \multirow{4}{*}{\makecell[l]{Llama3.2-3B\\-Instruct}}
      & Distill
      & \beststdinline{56.3}{0.8} & \stdinline{84.0}{0.5} & 0.75 & 4.86 & 7.58
      & \stdinline{26.0}{0.3} & \stdinline{48.5}{1.3} & 0.77 & 2.55 & 7.58 \\
      & DITTO
      & \stdinline{14.0}{0.5} & \stdinline{99.3}{0.0} & 0.51 & 5.72 & 4.72
      & \stdinline{1.00}{0.0} & \stdinline{99.5}{0.0} & 0.35 & 1.81 & 4.72 \\
      & DPO
      & \stdinline{0.75}{0.0} & \stdinline{16.0}{0.5} & 0.36 & 1.84 & 1.58
      & \stdinline{7.75}{0.3} & \stdinline{34.3}{0.8} & 0.60 & 1.52 & 1.58 \\
      & RLSpoofer
      & \stdinline{55.3}{1.5} & \stdinline{80.5}{2.5} & 0.76 & 6.12 & 6.53
      & \beststdinline{52.0}{1.2} & \stdinline{75.5}{3.0} & 0.72 & 4.60 & 6.53 \\
    \bottomrule
  \end{tabular}
\end{table}

\begin{table}[t]
  \centering
  \caption{Spoof Success Rate (SSR, \%), Spoof Rate (SR, \%), P-SP, Perplexity (PPL), and GPT-score (GPTS) across models. For readability, we report SSR and SR as mean $\pm$ standard deviation over two random seeds. \textbf{Bold} indicates the best mean SSR for each model under each watermarking scheme. Higher SSR indicates stronger spoofing performance.}
  \label{tab:pf_pmark_main}
  \small
  \setlength{\tabcolsep}{2pt}
  \renewcommand{\arraystretch}{0.95}
  \begin{tabular}{@{}llcccccccccc@{}}
    \toprule
    & & \multicolumn{5}{c}{\textbf{PF}} & \multicolumn{5}{c}{\textbf{PMark}} \\
    \cmidrule(lr){3-7} \cmidrule(lr){8-12}
    \textbf{Models} & \textbf{Methods}
    & \textbf{SSR} & \textbf{SR} & \textbf{P-SP} & \textbf{PPL} & \textbf{GPTS}
    & \textbf{SSR} & \textbf{SR} & \textbf{P-SP} & \textbf{PPL} & \textbf{GPTS} \\
    \midrule

    \multirow{4}{*}{Qwen3-0.6B}
      & Distill
      & \stdinline{6.50}{0.2} & \stdinline{7.25}{0.3} & 0.97 & 2.06 & 5.91
      & \stdinline{20.0}{0.5} & \stdinline{23.3}{0.8} & 0.90 & 2.31 & 5.91 \\
      & DITTO
      & \stdinline{5.50}{0.2} & \stdinline{11.5}{0.3} & 0.79 & 2.21 & 2.34
      & \stdinline{11.8}{0.3} & \stdinline{25.3}{0.5} & 0.63 & 2.38 & 2.34 \\
      & DPO
      & \stdinline{2.50}{0.2} & \stdinline{22.5}{0.4} & 0.57 & 1.56 & 5.31
      & \stdinline{6.25}{0.2} & \stdinline{14.5}{0.5} & 0.63 & 2.99 & 5.31 \\
      & RLSpoofer
      & \beststdinline{33.3}{0.3} & \stdinline{52.0}{0.2} & 0.66 & 2.19 & 5.28
      & \beststdinline{29.5}{0.4} & \stdinline{31.5}{0.3} & 0.92 & 2.03 & 5.80 \\
    \midrule

    \multirow{4}{*}{Qwen3-1.7B}
      & Distill
      & \stdinline{7.00}{0.3} & \stdinline{8.00}{0.3} & 0.96 & 2.06 & 7.24
      & \stdinline{20.3}{0.2} & \stdinline{22.5}{0.2} & 0.90 & 2.34 & 7.24 \\
      & DITTO
      & \stdinline{7.50}{0.4} & \stdinline{9.50}{0.2} & 0.86 & 2.27 & 4.43
      & \stdinline{16.5}{0.5} & \stdinline{30.5}{0.6} & 0.68 & 2.48 & 4.43 \\
      & DPO
      & \stdinline{4.25}{0.2} & \stdinline{17.0}{0.4} & 0.58 & 2.31 & 9.12
      & \stdinline{22.5}{0.2} & \stdinline{24.0}{0.3} & 0.93 & 2.32 & 9.12 \\
      & RLSpoofer
      & \beststdinline{29.0}{0.3} & \stdinline{38.5}{0.3} & 0.73 & 4.10 & 6.21
      & \beststdinline{29.5}{0.5} & \stdinline{31.3}{0.4} & 0.90 & 2.29 & 6.93 \\
    \midrule

    \multirow{4}{*}{Qwen3-4B}
      & Distill
      & \stdinline{6.00}{0.2} & \stdinline{6.50}{0.2} & 0.96 & 2.11 & 7.67
      & \stdinline{21.5}{0.3} & \stdinline{27.3}{0.5} & 0.92 & 2.39 & 7.67 \\
      & DITTO
      & \stdinline{3.50}{0.4} & \stdinline{6.25}{0.2} & 0.88 & 2.47 & 6.01
      & \stdinline{16.5}{0.2} & \stdinline{23.5}{0.2} & 0.71 & 2.50 & 6.01 \\
      & DPO
      & \stdinline{5.25}{0.2} & \stdinline{8.75}{0.3} & 0.88 & 2.33 & 9.30
      & \stdinline{17.5}{0.4} & \stdinline{30.8}{0.5} & 0.68 & 5.06 & 9.30 \\
      & RLSpoofer
      & \beststdinline{62.0}{0.7} & \stdinline{82.8}{0.5} & 0.77 & 2.58 & 6.60
      & \beststdinline{36.3}{0.3} & \stdinline{40.8}{0.4} & 0.91 & 2.05 & 7.25 \\
    \midrule

    \multirow{4}{*}{\makecell[l]{Qwen2.5-3B\\-Instruct}}
      & Distill
      & \stdinline{6.50}{0.2} & \stdinline{8.75}{0.4} & 0.93 & 2.04 & 7.43
      & \stdinline{22.3}{0.4} & \stdinline{26.0}{0.4} & 0.91 & 2.01 & 7.43 \\
      & DITTO
      & \stdinline{5.25}{0.2} & \stdinline{12.3}{0.2} & 0.72 & 2.13 & 4.67
      & \stdinline{11.3}{0.2} & \stdinline{29.3}{0.5} & 0.56 & 2.00 & 4.67 \\
      & DPO
      & \stdinline{6.25}{0.2} & \stdinline{12.3}{0.4} & 0.83 & 1.98 & 7.15
      & \stdinline{24.3}{0.3} & \stdinline{30.3}{0.3} & 0.95 & 2.22 & 7.15 \\
      & RLSpoofer
      & \beststdinline{50.3}{0.3} & \stdinline{79.5}{0.3} & 0.68 & 1.79 & 6.61
      & \beststdinline{30.3}{0.2} & \stdinline{33.3}{0.3} & 0.89 & 1.94 & 7.41 \\
    \midrule

    \multirow{4}{*}{\makecell[l]{Llama3.2-3B\\-Instruct}}
      & Distill
      & \stdinline{8.75}{0.4} & \stdinline{9.75}{0.2} & 0.93 & 4.49 & 7.58
      & \stdinline{23.3}{0.6} & \stdinline{27.8}{0.5} & 0.89 & 2.22 & 7.58 \\
      & DITTO
      & \stdinline{6.50}{0.2} & \stdinline{11.3}{0.2} & 0.79 & 4.90 & 4.72
      & \stdinline{18.0}{0.4} & \stdinline{34.2}{0.5} & 0.65 & 2.03 & 4.72 \\
      & DPO
      & \stdinline{6.25}{0.2} & \stdinline{11.3}{0.4} & 0.67 & 2.00 & 1.58
      & \stdinline{25.0}{0.3} & \stdinline{30.3}{0.4} & 0.87 & 1.94 & 1.58 \\
      & RLSpoofer
      & \beststdinline{49.8}{0.4} & \stdinline{60.5}{0.4} & 0.85 & 8.99 & 6.53
      & \beststdinline{33.3}{0.2} & \stdinline{34.3}{0.3} & 0.92 & 2.23 & 7.01 \\
    \bottomrule
  \end{tabular}
\end{table}

As shown in Tables~\ref{tab:ewd_sweet_main}, \ref{tab:kgw_unigram_main} and \ref{tab:pf_pmark_main}, we observe that RLSpoofer achieves the best or second-best SSR in most settings while maintaining competitive rephrasing quality. In contrast, DITTO often attains high SR but much lower P-SP, suggesting that many of its successful detections come from semantically distorted rewrites. DPO shows the opposite pattern: it typically preserves semantics better, but fails to induce sufficient distribution shift toward the target watermark, resulting in consistently weak SSR. Distill is competitive on several logit-based schemes, but is substantially less effective on the more challenging sampling-based settings.

For the logit-based watermarks, including EWD, SWEET, KGW, and Unigram, we find that these schemes are already vulnerable to distribution-matching baselines, although RLSpoofer remains the most consistent method overall. In particular, RLSpoofer achieves SSR above 50\% on EWD and SWEET for nearly all attacker models, indicating that it can preserve semantics while aligning with watermark-favored distributions. On KGW, Distill is competitive in some cases, suggesting that its watermark signal is relatively learnable from large corpora. On Unigram, however, the gap becomes much larger: while Distill and DITTO degrade substantially, RLSpoofer consistently maintains strong spoofing performance.

The largest gap appears on the sampling-based watermarks PF and PMark. We observe that all three baselines struggle on PF, with uniformly low SSR across models, indicating that they largely fail to reproduce the watermark signal. By contrast, RLSpoofer achieves substantially stronger spoofing performance, most notably reaching 62\% SSR on PF with Qwen3-4B. This result is especially notable because PF is distortion-free in expectation, making its signal difficult to capture through standard distillation-based attacks. On PMark, RLSpoofer again achieves the best SSR across all attacker models while maintaining high P-SP, showing that its gain does not come merely from sacrificing rephrasing quality.

In conclusion, we observe that logit-based watermarks are vulnerable to attacks that learn their induced distributional bias, but RLSpoofer is more reliable because it better preserves semantics while shifting the output distribution. Moreover, sampling-based watermarks, especially PF, appear much more robust to existing baselines, yet remain highly vulnerable to RLSpoofer, suggesting that standard evaluation pipelines can substantially underestimate spoofing risk.

Additionally, we present example visualizations for the logit-based watermarking algorithms EWD, SWEET, KGW, and Unigram, comparing the watermarked rewrites produced by RLSpoofer with the corresponding original unwatermarked texts. The results are shown in Figure~\ref{fig:wm_ewd} to Figure~\ref{fig:uwm_Unigram}, with paired figures showing high semantic similarities (P-SP score>0.9).

\begin{figure}
    \centering
    \includegraphics[width=\linewidth]{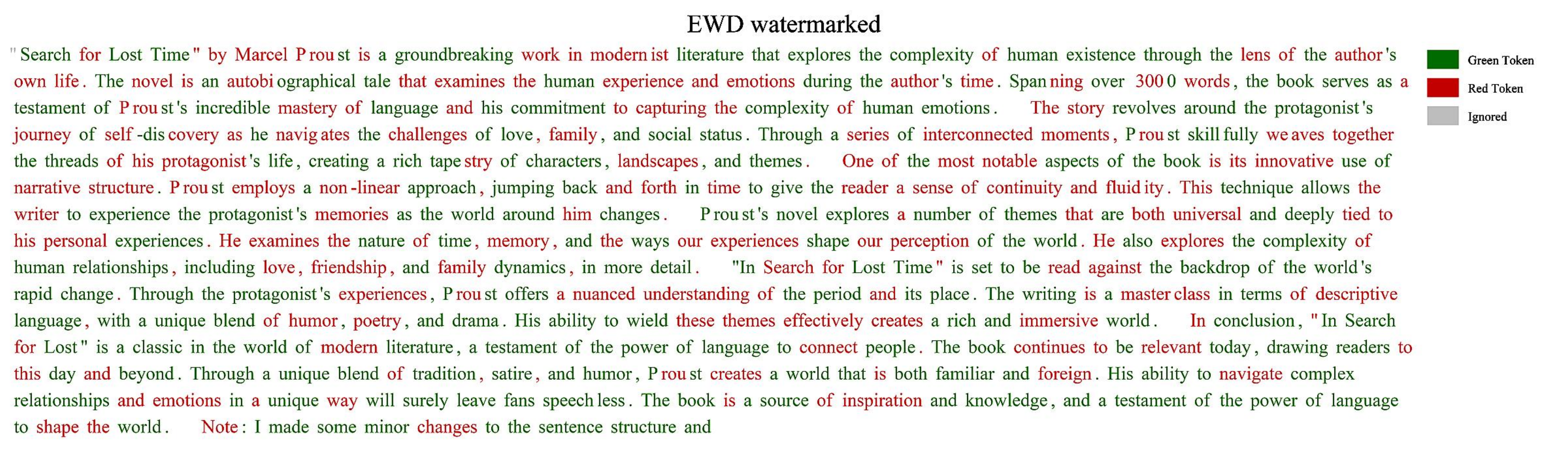}
    \caption{Watermarked EWD example}
    \label{fig:wm_ewd}
\end{figure}
\begin{figure}
    \centering
    \includegraphics[width=\linewidth]{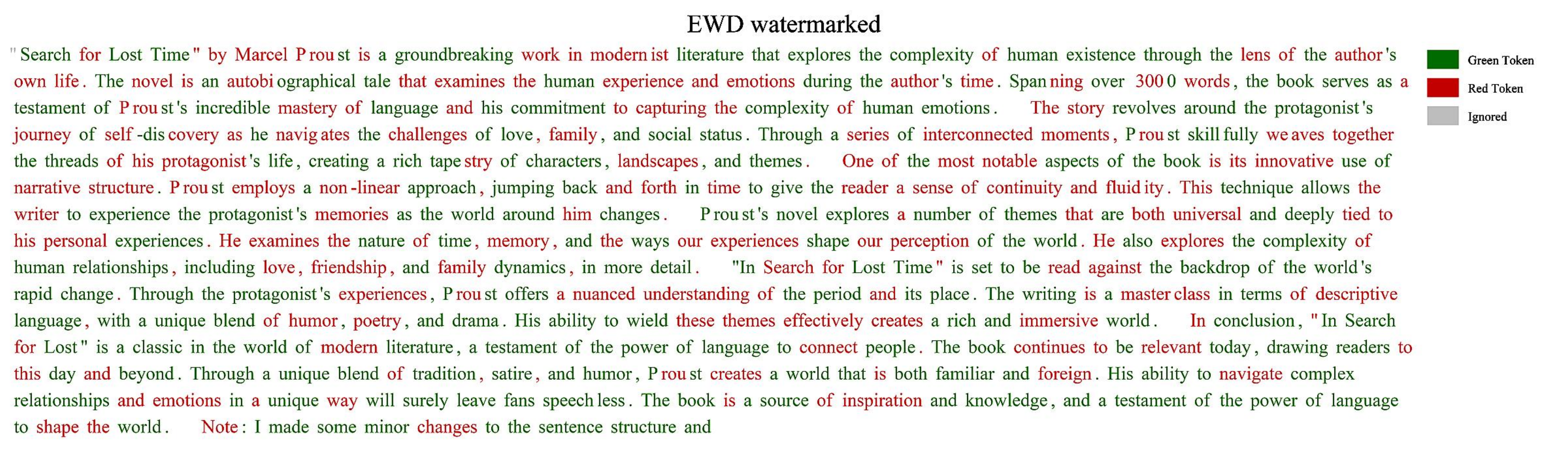}
    \caption{Unwatermarked EWD example}
    \label{fig:uwm_ewd}
\end{figure}
\begin{figure}
    \centering
    \includegraphics[width=\linewidth]{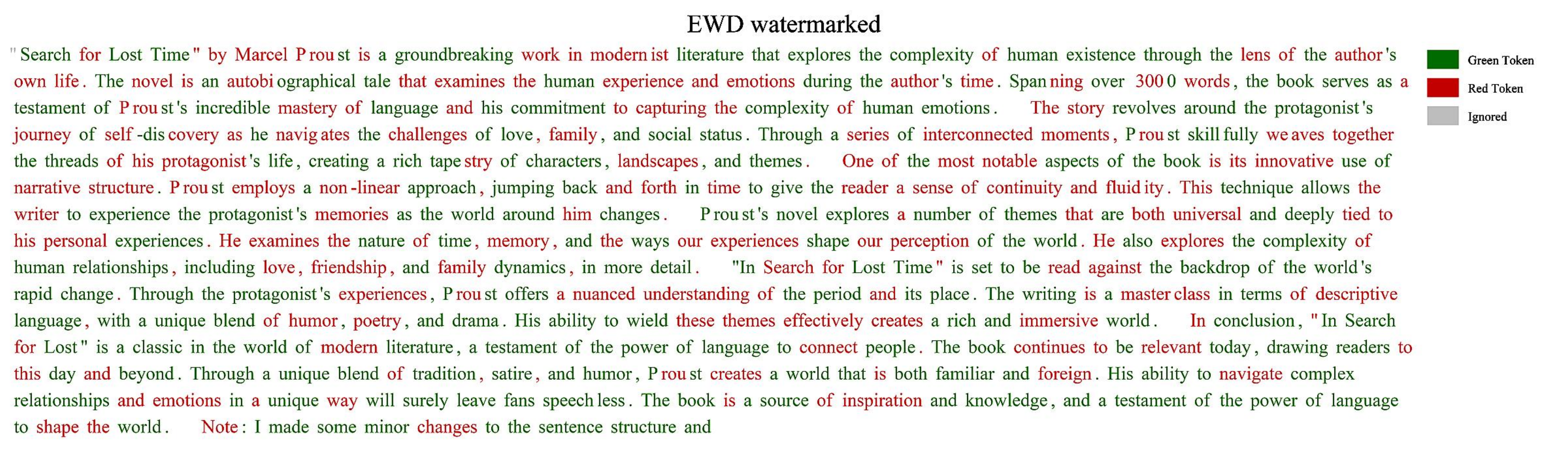}
    \caption{Watermarked EWD example}
    \label{fig:wm_SWEET}
\end{figure}
\begin{figure}
    \centering
    \includegraphics[width=\linewidth]{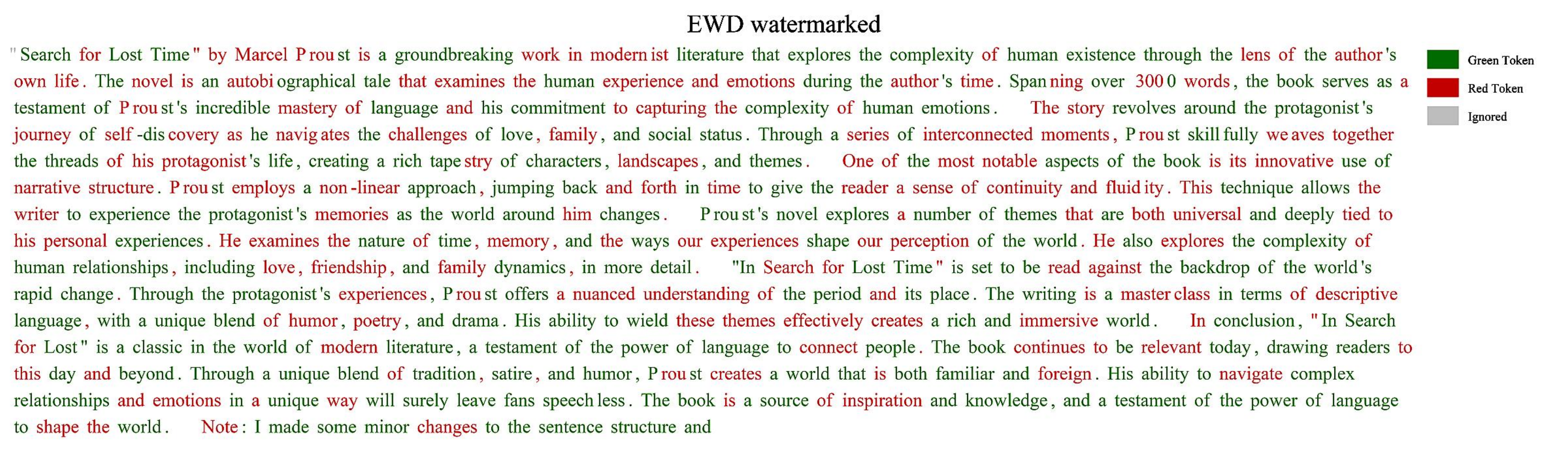}
    \caption{Unwatermarked EWD example}
    \label{fig:uwm_SWEET}
\end{figure}
\begin{figure}
    \centering
    \includegraphics[width=\linewidth]{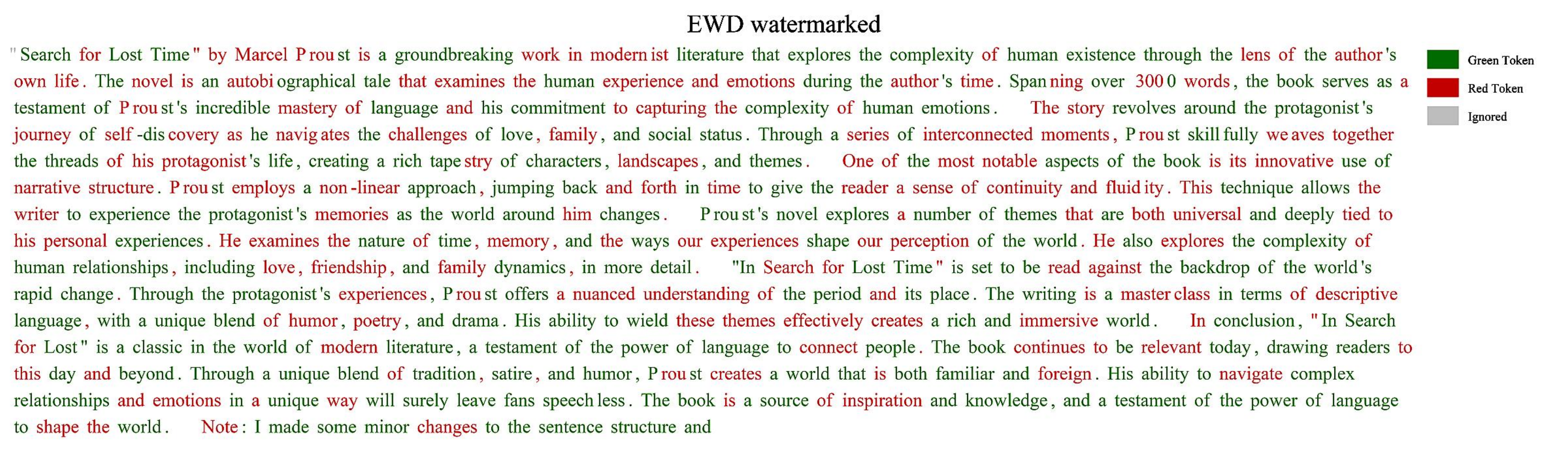}
    \caption{Watermarked EWD example}
    \label{fig:wm_KGW}
\end{figure}
\begin{figure}
    \centering
    \includegraphics[width=\linewidth]{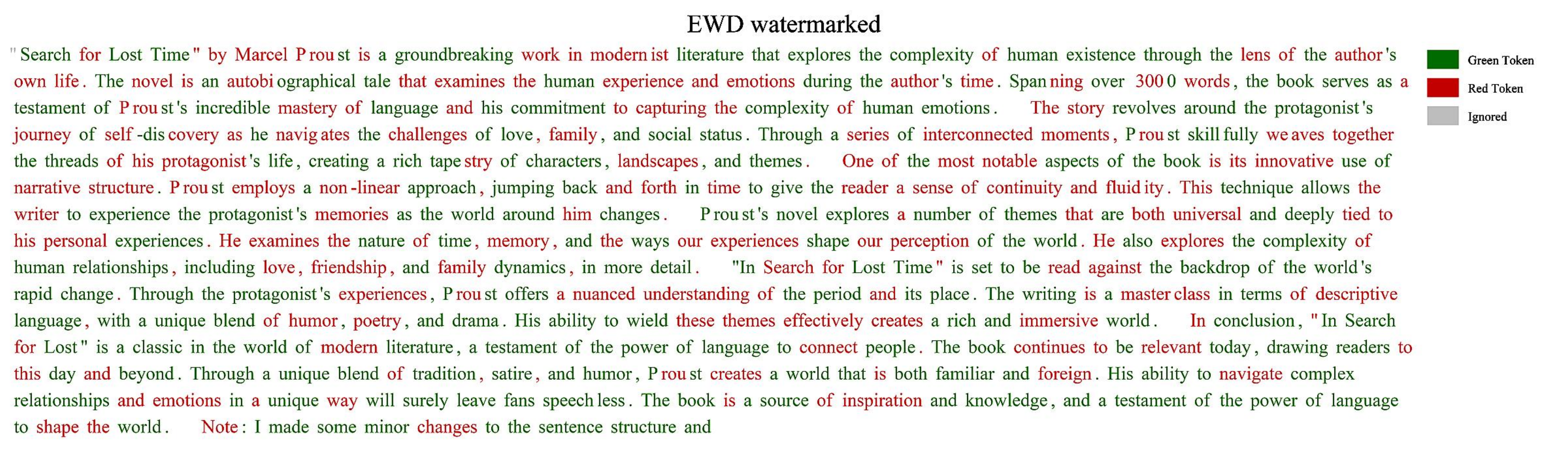}
    \caption{Unwatermarked EWD example}
    \label{fig:uwm_KGW}
\end{figure}
\begin{figure}
    \centering
    \includegraphics[width=\linewidth]{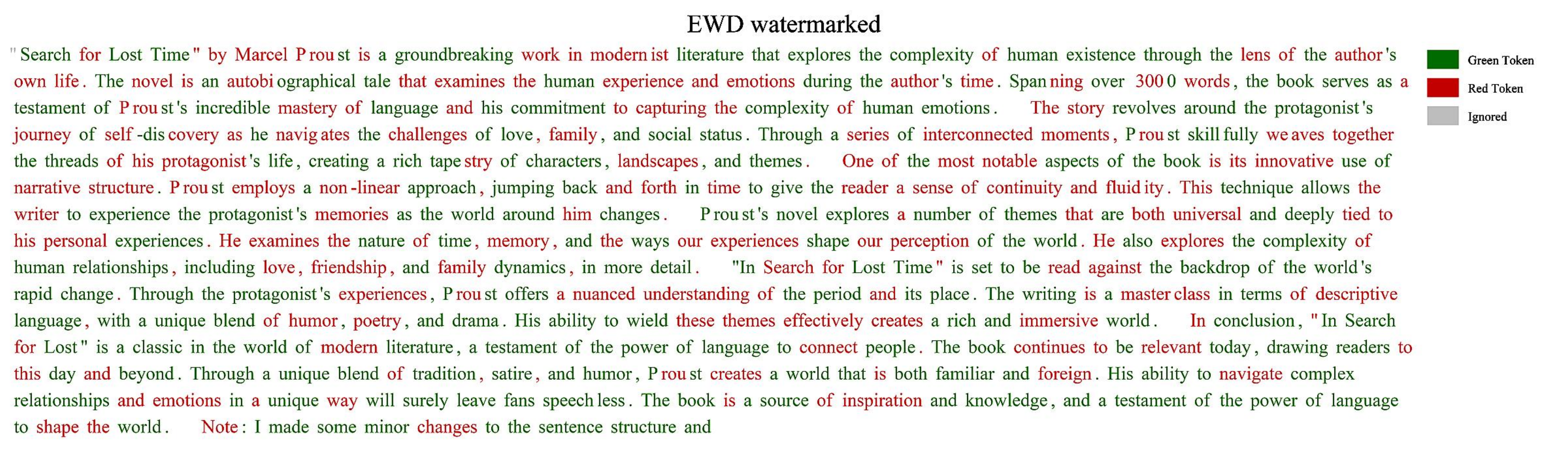}
    \caption{Watermarked EWD example}
    \label{fig:wm_Unigram}
\end{figure}
\begin{figure}
    \centering
    \includegraphics[width=\linewidth]{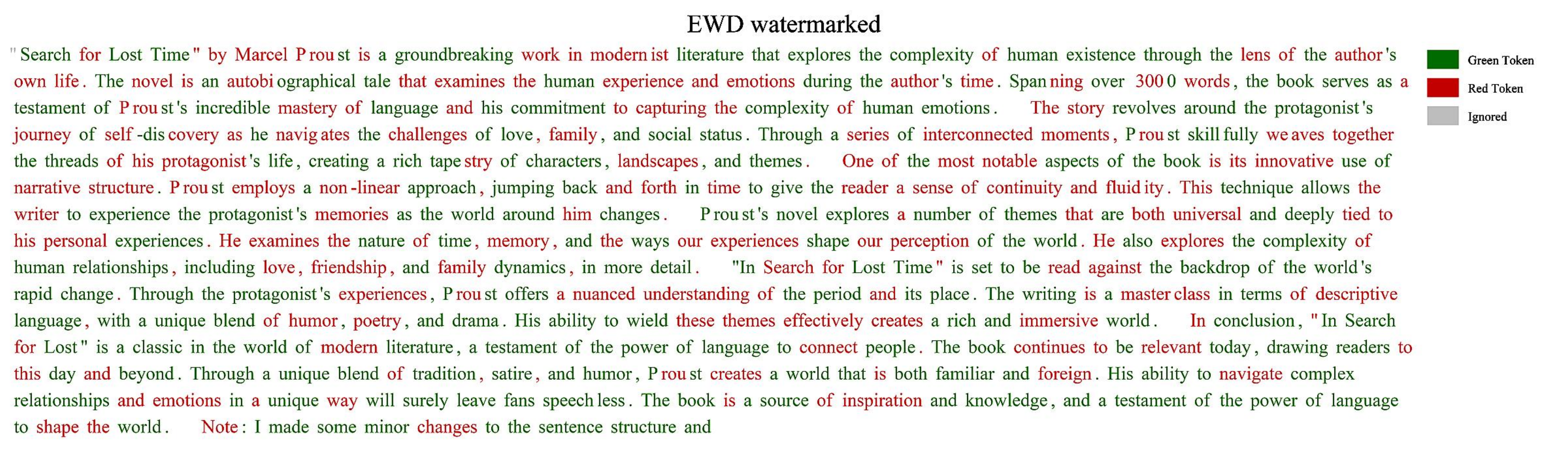}
    \caption{Unwatermarked EWD example}
    \label{fig:uwm_Unigram}
\end{figure}

\subsection{Detection Score Distribution}
\label{app:detection_scores}
The detection scores in Figure~\ref{fig:detection_pos}(a) demonstrate that RLSpoofer effectively shifts the distribution of rephrased texts toward the watermarked distribution and away from the unwatermarked reference. To establish the unwatermarked distribution, we directly use the test set's z-scores; for the watermarked distribution, we use test-set rewrites generated by the watermarked Llama-3.1-8B-Instruct. Finally, the rephrased distribution consists of detection scores from outputs paraphrased by our RLSpoofer-trained Qwen3-4B. The clear separation between the RLSpoofer-rephrased and human-like distributions confirms the attack's effectiveness at stealing watermarks, exposing a fundamental vulnerability in current schemes.

\subsection{Effect of Local capacity mass}
\label{app:local_capacity_mass}

To more directly illustrate the effectiveness of local capacity mass in identifying feasible room for watermark injection during training, we conduct an analysis using the first checkpoint obtained by training Qwen3-4B on the 100-sample SWEET watermark training set. Specifically, we first use this checkpoint to rewrite all 100 training samples and collect the resulting paraphrased outputs. Based on these outputs, we compute the raw token-level reward signal, namely the unweighted log-likelihood ratio \(\log \frac{P_{wm}(x_t' \mid h_t)}{P_h(x_t' \mid h_t)}\), together with the corresponding local capacity mass under the human-like surrogate distribution \(P_h\), which is approximated by the original Qwen3-4B reference model. We then compute the final reward used in RL training, \(r_t\), and average the raw reward, the capacity mass, and the final reward over all token occurrences across the 100 samples. After grouping tokens by part-of-speech tags, we obtain the statistics shown in Figure~\ref{fig:detection_pos}(b).

As shown in Figure~\ref{fig:detection_pos}(b), the induced reward is small on numerals, which are typically critical for preserving semantic fidelity, but larger on adjectives, verbs, and adverbs, which admit greater lexical flexibility. This suggests that local capacity mass suppresses overly sharp rewards on semantically rigid tokens while amplifying rewards on more substitutable ones. Therefore, rather than uniformly favoring all watermark-preferred tokens, RLSpoofer concentrates watermark injection on positions with genuinely feasible semantic slack.

Additionally, we verify that this weighting is not only intuitively aligned with semantic flexibility, but also critical to spoofing performance. Specifically, we replace the original weight \(1-p_{\max}\) with either a uniform weight of \(1\) or the reverse weight \(p_{\max}\). We use the same hyperparameter settings as in Table~\ref{tab:hyperparamsRLS} and report the detailed results in Table~\ref{tab:reward_weight_ablation_app}. As discussed in the main text, both replacements consistently degrade spoofing performance, indicating that effective watermark spoofing requires identifying positions with sufficient semantics-preserving redistribution room, rather than uniformly enlarging token-level rewards.
\begin{table}[ht]
\centering
\small
\setlength{\tabcolsep}{6pt}
\renewcommand{\arraystretch}{0.95}
\caption{Results across different reward weighting.}
\begin{tabular}{llcccccc}
\toprule
& & \multicolumn{3}{c}{\textbf{EWD}} & \multicolumn{3}{c}{\textbf{SWEET}} \\
\cmidrule(lr){3-5} \cmidrule(lr){6-8}
\textbf{Model} & \textbf{Weight} & \textbf{SSR} & \textbf{SR} & \textbf{P-SP} & \textbf{SSR} & \textbf{SR} & \textbf{P-SP} \\
\midrule
\multirow{3}{*}{Qwen3-0.6B}
& $1-p_{\max}$ & 54.3 & 80.5 & 0.73 & 50.5 & 66.3 & 0.79 \\
& 1            & 42.8 & 67.5 & 0.71 & 39.0 & 68.8 & 0.71 \\
& $p_{\max}$   & 40.5 & 73.0 & 0.67 & 29.8 & 41.0 & 0.75 \\
\midrule
\multirow{3}{*}{Qwen3-4B}
& $1-p_{\max}$ & 56.5 & 87.0 & 0.73 & 52.3 & 71.3 & 0.75 \\
& 1            & 35.5 & 64.3 & 0.68 & 28.8 & 38.3 & 0.78 \\
& $p_{\max}$   & 28.0 & 47.3 & 0.70 & 25.0 & 43.0 & 0.68 \\
\bottomrule
\end{tabular}
\label{tab:reward_weight_ablation_app}
\end{table}

\subsection{Effectiveness of conservative semantic rewards}
\label{app:semantic_rewards}

We further study how the design of the sequence-level semantic reward affects spoofing performance. Specifically, we compare four variants: {Min.}, {Avg.}, {Hum.}, and {W.M.}, where {Min.} denotes the minimum of the semantic similarities to the human-written text and the watermarked rewrite, {Avg.} denotes their average, and {Hum.} and {W.M.} use only the human-written text or the watermarked rewrite as the semantic reference, respectively. We use the same training hyperparameter settings reported in Table~\ref{tab:hyperparamsRLS}, the results are shown in Table~\ref{tab:semantic_reward_ablation}.

We observe that {Min.} consistently achieves the best SSR across both attacker models and all three watermarking schemes. For Qwen3-0.6B, {Min.} yields the highest SSR and SR on EWD, SWEET, and PF, outperforming the other reward variants by a clear margin. A similar trend holds for Qwen3-4B, where {Min.} again gives the best SSR and SSR, most notably improving PF from 50.3\% and 51\% under {Avg.} and {W.M.} to 62\%. These results suggest that a conservative semantic reward provides a better optimization target for watermark spoofing. We attribute this improvement to the fact that successful spoofing must simultaneously preserve the meaning of the original text and remain close to the watermarked rewrite, which serves as the surrogate target distribution. Using only one reference may encourage one-sided alignment, while averaging the two scores can still mask weak alignment to one side. In contrast, {Min.} explicitly enforces both constraints and therefore yields more reliable semantic control during training.

\begin{table}[ht]
\centering
\small
\setlength{\tabcolsep}{5pt}
\renewcommand{\arraystretch}{0.95}
\caption{Ablation on different semantic reward designs.}
\begin{tabular}{llccccccccc}
\toprule
& & \multicolumn{3}{c}{\textbf{EWD}} & \multicolumn{3}{c}{\textbf{SWEET}} & \multicolumn{3}{c}{\textbf{PF}} \\
\cmidrule(lr){3-5} \cmidrule(lr){6-8} \cmidrule(lr){9-11}
\textbf{Model} & \textbf{Sem.R.} & \textbf{SSR} & \textbf{SR} & \textbf{P-SP} & \textbf{SSR} & \textbf{SR} & \textbf{P-SP} & \textbf{SSR} & \textbf{SR} & \textbf{P-SP} \\
\midrule
\multirow{4}{*}{Qwen3-0.6B}
& Min. & 54.3 & 80.5 & 0.73 & 50.5 & 66.3 & 0.79 & 33.3 & 52.0 & 0.66 \\
& Avg. & 44.8 & 78.3 & 0.70 & 48.8 & 60.5 & 0.77 & 26.3 & 42.3 & 0.72 \\
& Hum. & 42.3 & 69.0 & 0.70 & 43.0 & 59.5 & 0.75 & 24.5 & 40.0 & 0.76 \\
& W.M. & 48.5 & 71.3 & 0.70 & 46.8 & 61.3 & 0.76 & 25.5 & 51.3 & 0.64 \\
\midrule
\multirow{4}{*}{Qwen3-4B}
& Min. & 56.5 & 87.0 & 0.73 & 52.3 & 71.3 & 0.75 & 62.0 & 82.8 & 0.77 \\
& Avg. & 47.3 & 78.0 & 0.70 & 34.3 & 49.0 & 0.76 & 50.3 & 76.3 & 0.78 \\
& Hum. & 47.5 & 80.0 & 0.68 & 47.0 & 59.3 & 0.76 & 48.5 & 69.3 & 0.76 \\
& W.M. & 45.5 & 62.3 & 0.75 & 48.5 & 61.8 & 0.78 & 51.0 & 70.8 & 0.75 \\
\bottomrule
\end{tabular}
\label{tab:semantic_reward_ablation}
\end{table}

\subsection{Guidance of cross-entropy anchor}
\label{app:ceanchor}

We examine the role of the cross-entropy (CE) anchor in stabilizing RLSpoofer training. We compare the RLSpoofer method ({With Anchor}) with two alternatives: removing the CE anchor ({Without Anchor}) and replacing RL training with supervised distillation on the same 100 training pairs ({Distillation (100)}). We adopt the training hyperparameters reported in Appendix~\ref{app:experimental_setup}. Specifically, for \textit{Without Anchor} setting, we set the learning rate to $1\times10^{-4}$ to stabilize the training process. The results are reported in Table~\ref{tab:ce_anchor_ablation}.

We find that the CE anchor is critical for stable and effective spoofing. Removing it leads to substantial drops in SSR across all settings. For example, on Qwen3-4B, SSR decreases from 56.5\% to 33.5\% on EWD, and we observe similar degradations across all other settings. This shows that, even with carefully designed token-level and semantic rewards, unconstrained policy optimization can drift away from the desired watermarked rewriting distribution.  Moreover, supervised distillation on the same 100 training pairs performs extremely poorly, with nearly zero SSR in all settings despite relatively high P-SP. This indicates that the gain of RLSpoofer does not come merely from imitation on limited data, but from the combination of RL-based distributional alignment and CE-based stabilization. Overall, these results confirm that the CE anchor is essential for maintaining a useful optimization trajectory while still allowing the policy to shift toward watermark-favored generations.

\begin{table}[ht]
\centering
\small
\setlength{\tabcolsep}{6pt}
\caption{Ablation on the effect of the cross-entropy anchor.}
\begin{tabular}{llccccccccc}
\toprule
& & \multicolumn{3}{c}{\textbf{With Anchor}} & \multicolumn{3}{c}{\textbf{Without Anchor}} & \multicolumn{3}{c}{\textbf{Distillation (100)}} \\
\cmidrule(lr){3-5} \cmidrule(lr){6-8} \cmidrule(lr){9-11}
\textbf{Model} & \textbf{Scheme} & \textbf{SSR} & \textbf{SR} & \textbf{P-SP} & \textbf{SSR} & \textbf{SR} & \textbf{P-SP} & \textbf{SSR} & \textbf{SR} & \textbf{P-SP} \\
\midrule
\multirow{3}{*}{Qwen3-0.6B}
& EWD   & 54.3 & 80.5 & 0.73 & 29.3 & 45.5 & 0.69 & 0.25 & 0.50 & 0.86 \\
& SWEET & 50.5 & 66.3 & 0.79 & 27.3 & 50.0 & 0.70 & 0.25 & 0.25 & 0.84 \\
& PF    & 33.3 & 52.0 & 0.66 & 12.3 & 28.3 & 0.72 & 0.00 & 0.25 & 0.87 \\
\midrule
\multirow{3}{*}{Qwen3-4B}
& EWD   & 56.5 & 87.0 & 0.73 & 33.5 & 50.3 & 0.73 & 0.00 & 0.25 & 0.93 \\
& SWEET & 52.3 & 71.3 & 0.75 & 26.8 & 44.3 & 0.78 & 0.25 & 0.25 & 0.94 \\
& PF    & 62.0 & 82.8 & 0.77 & 25.5 & 45.3 & 0.70 & 0.00 & 0.25 & 0.84 \\
\bottomrule
\end{tabular}
\label{tab:ce_anchor_ablation}
\end{table}

\subsection{Training set size ablation}
\label{app:training_set_size}

We conduct an empirical study to examine the sensitivity of RLSpoofer to training set size. Specifically, we train the models on data generated by watermarked Llama3.1-8B-Instruct under three watermarking schemes, EWD, SWEET, and PF, using training sets containing 50, 100, and 200 samples, where each sample is standardized to 500 tokens. To ensure a comparable number of update steps across different training set sizes, we train for 20 epochs when using 50 samples and for 5 epochs when using 200 samples. For the 100-sample setting, we use the standard training schedule described in Table~\ref{tab:hyperparamsRLS}. We adopt the same remaining hyperparameters and report the results in Table~\ref{tab:training_size_ablation}.

We observe that RLSpoofer achieves strong spoofing performance with only 50 training samples. On the logit-based watermarks EWD and SWEET, increasing the training set generally improves SSR for both attacker models, with the best performance typically obtained using 200 samples. In contrast, on the sampling-based distortion-free PF watermark, enlarging the training set does not lead to consistent gains. For both Qwen3-0.6B and Qwen3-4B, the best PF performance is achieved with 100 samples, while using 200 samples leads to a noticeable drop in SSR. These results suggest that RLSpoofer is sample-efficient and that relatively small training sets are already sufficient to expose strong spoofing vulnerabilities, especially for the more challenging distortion-free setting.

\begin{table}[ht]
\centering
\small
\setlength{\tabcolsep}{6pt}
\caption{Effect of training set size on spoofing performance.}
\begin{tabular}{llccccccccc}
\toprule
& & \multicolumn{3}{c}{\textbf{EWD}} & \multicolumn{3}{c}{\textbf{SWEET}} & \multicolumn{3}{c}{\textbf{PF}} \\
\cmidrule(lr){3-5} \cmidrule(lr){6-8} \cmidrule(lr){9-11}
\textbf{Model} & \textbf{Samples} & \textbf{SSR} & \textbf{SR} & \textbf{P-SP} & \textbf{SSR} & \textbf{SR} & \textbf{P-SP} & \textbf{SSR} & \textbf{SR} & \textbf{P-SP} \\
\midrule
\multirow{3}{*}{Qwen3-0.6B}
& 50  & 44.3 & 77.0 & 0.72 & 42.5 & 65.5 & 0.74 & 12.3 & 19.3 & 0.72 \\
& 100 & 54.3 & 80.5 & 0.73 & 50.5 & 66.3 & 0.79 & 33.3 & 52.0 & 0.66 \\
& 200 & 56.8 & 85.0 & 0.74 & 51.3 & 71.5 & 0.79 & 20.5 & 37.0 & 0.76 \\
\midrule
\multirow{3}{*}{Qwen3-4B}
& 50  & 46.5 & 79.0 & 0.73 & 43.0 & 72.0 & 0.70 & 20.8 & 44.0 & 0.68 \\
& 100 & 56.5 & 87.0 & 0.73 & 52.3 & 71.3 & 0.75 & 62.0 & 82.8 & 0.77 \\
& 200 & 58.5 & 88.0 & 0.74 & 54.3 & 65.3 & 0.77 & 25.3 & 36.3 & 0.81 \\
\bottomrule
\end{tabular}
\label{tab:training_size_ablation}
\end{table}

\subsection{Sensitivity to surrogate model selection}
\label{app:sensitivity}

We examine the sensitivity of RLSpoofer to the choice of surrogate reference model. Specifically, for each attacker, we compare two surrogate choices for approximating the target distributions: Qwen3-0.6B and the attacker itself ({Self}). The results are reported in Table~\ref{tab:surrogate_choice_ablation}. We find that surrogate choice has a limited impact on Qwen3-4B, where the two settings yield similar performance on EWD and SWEET, and {Self} performs slightly better on PF. In contrast, Qwen3-8B is much more sensitive: using Qwen3-0.6B consistently gives substantially higher SSR than {Self} across all three watermarking schemes. In particular, SSR improves from 45.3\% to 60.5\% on EWD, from 35\% to 54.3\% on SWEET, and from 29.8\% to 58.5\% on PF. These results suggest that surrogate selection becomes increasingly important as attacker capacity grows.

\begin{table}[ht]
\centering
\small
\setlength{\tabcolsep}{6pt}
\caption{Sensitivity of RLSpoofer to surrogate model choice.}
\begin{tabular}{llcccccc}
\toprule
& & \multicolumn{3}{c}{\textbf{Qwen3-0.6B}} & \multicolumn{3}{c}{\textbf{Self}} \\
\cmidrule(lr){3-5} \cmidrule(lr){6-8}
\textbf{Attacker} & \textbf{Scheme} & \textbf{SSR} & \textbf{SR} & \textbf{P-SP} & \textbf{SSR} & \textbf{SR} & \textbf{P-SP} \\
\midrule
\multirow{3}{*}{Qwen3-4B}
& EWD   & 56.5 & 85.8 & 0.73 & 56.5 & 87.0 & 0.73 \\
& SWEET & 53.3 & 73.0 & 0.75 & 52.3 & 71.3 & 0.75 \\
& PF    & 56.5 & 83.8 & 0.76 & 62.0 & 82.8 & 0.77 \\
\midrule
\multirow{3}{*}{Qwen3-8B}
& EWD   & 60.5 & 78.0 & 0.75 & 45.3 & 67.8 & 0.71 \\
& SWEET & 54.3 & 75.5 & 0.75 & 35.0 & 65.5 & 0.68 \\
& PF    & 58.5 & 79.3 & 0.78 & 29.8 & 33.8 & 0.86 \\
\bottomrule
\end{tabular}
\label{tab:surrogate_choice_ablation}
\end{table}

\subsection{Cross watermark transferability of RLSpoofer}
\label{app:transferability}
We further evaluate the zero-shot cross-watermark transferability of RLSpoofer. Specifically, we directly reuse the RLSpoofer-trained Qwen3-4B models obtained in Section~\ref{sec:main_Results_main}, and test each model on the other watermarking schemes without any further tuning. The results are reported in Table~\ref{tab:cross_watermark_transfer}.

We observe that transferability is not strictly constrained by the watermark family. In particular, the KGW-style logit-based watermarks EWD and SWEET exhibit substantial bidirectional transfer: training on EWD achieves 50.0\% SSR on SWEET, while training on SWEET achieves 47.8\% SSR on EWD. By contrast, interactions involving the sampling-based PF watermark are notably directional. Training on PF transfers effectively to EWD, achieving 52.5\% SSR, whereas the reverse direction is much weaker, with EWD-to-PF achieving only 6.0\% SSR. Similarly, PF-to-SWEET almost completely fails, yielding only 0.25\% SSR. These results suggest that transferability is determined not only by broad watermark family similarity, but also by more intricate and directional overlaps in the vulnerabilities induced by different watermarking schemes. 
\begin{table}[ht]
\centering
\small
\setlength{\tabcolsep}{6pt}
\caption{Cross-watermark transferability of RLSpoofer on Qwen3-4B.}
\begin{tabular}{lccccccccc}
\toprule
\textbf{Test}& \multicolumn{3}{c}{\textbf{EWD}} & \multicolumn{3}{c}{\textbf{SWEET}} & \multicolumn{3}{c}{\textbf{PF}} \\
\cmidrule(lr){1-1} \cmidrule(lr){2-4} \cmidrule(lr){5-7} \cmidrule(lr){8-10}
\textbf{Train} & \textbf{SSR} & \textbf{SR} & \textbf{P-SP} & \textbf{SSR} & \textbf{SR} & \textbf{P-SP} & \textbf{SSR} & \textbf{SR} & \textbf{P-SP} \\
\midrule
EWD   & 56.5 & 87.0 & 0.73 & 50.0 & 64.8 & 0.76 &  6.00 & 10.0 & 0.78 \\
SWEET & 47.8 & 58.8 & 0.78 & 52.3 & 71.3 & 0.75 & 47.0 & 70.5 & 0.72 \\
PF    & 52.5 & 80.5 & 0.75 &  0.25 &  1.00 & 0.61 & 62.0 & 82.8 & 0.77 \\
\bottomrule
\end{tabular}
\label{tab:cross_watermark_transfer}

\end{table}

\subsection{Limitations and Border Impact.}
\label{app:limitation}

\textbf{Limitations.} Although RLSpoofer demonstrates efficacy in evaluating the resilience of watermarks against spoofing attacks, it exhibits several notable limitations. Primarily, our method does not directly optimize the true detector-level spoofing objective. Instead, it relies on a tractable distributional surrogate together with reference-model-based local preference proxies for human-like and watermark-conditioned generation. Accordingly, our theory should be interpreted as motivating the reward design rather than providing a detector-agnostic guarantee of spoof success. In addition, the practical effectiveness of the method depends on the quality of the surrogate signals and can be sensitive to design choices such as reward weights, cross-entropy anchoring, and the choice of reference model.

\textbf{Border Impact.} Despite these limitations, RLspoofer provides a lightweight and practical stress test for evaluating whether current watermarking schemes remain reliable under realistic black-box attacks, which may help the community design more robust provenance and detection mechanisms. At the same time, because our method can expose concrete weaknesses in existing watermarking systems, it also carries dual-use risk if deployed irresponsibly. To mitigate this concern, we focus on evaluation rather than misuse, provide sufficient implementation details for scientific assessment.

\subsection{Responsible Release}
\label{app:safe_guard}
This work studies an attack method for evaluating the spoofing resilience of LLM watermarking schemes, and therefore carries potential dual-use risk. To mitigate misuse, we do not release the attack code during the review period. We provide implementation details, hyperparameter settings, and references sufficient for scientific assessment, while withholding a directly reusable attack pipeline. If the paper is accepted, we will carefully consider the form and scope of release to balance research transparency with misuse risk. Our goal is to support the evaluation of watermark robustness and the development of more spoofing-resistant defenses, rather than facilitate malicious deployment.

\section{Licenses of Existing Assets}
\label{app:asset_licenses}

We summarize the existing external assets used in this work, including codebases,
datasets, and models, together with their publicly stated licenses or access terms.
We cite the original papers and official repositories or model cards throughout the paper.
When an asset is distributed under a model-specific license, we use it in accordance
with the corresponding license agreement and acceptable use policy. The results are presented in Table~\ref{tab:asset_licenses}.

\begin{table}[ht]
\centering
\small
\setlength{\tabcolsep}{5pt}
\caption{Existing assets used in this work and their publicly stated licenses or access terms.}
\begin{tabular}{@{}lll@{}}
\toprule
Asset & Type & License / Terms of Use \\
\midrule

PMark~\cite{huo2025pmark}
& Code
& \makecell[l]{No explicit standard license found;\\ research reproduction only} \\

MarkLLM~\cite{pan2024markllm}
& Code
& Apache-2.0 \\

LLaMA-Factory~\cite{zheng2024llamafactory}
& Code
& Apache-2.0 \\

DITTO~\cite{an2025ditto}
& Code
& Apache-2.0 \\

TRL~\cite{vonwerra2022trl}
& Code
& Apache-2.0 \\

OpenR1~\cite{openr1}
& Code
& Apache-2.0 \\

C4 (RealNewslike subset)~\cite{raffel2020exploring}
& Dataset
& ODC-BY \\

Reddit WritingPrompts~\cite{verma2024ghostbuster}
& Dataset
& CC BY 3.0 \\

LFQA~\cite{krishna2023paraphrasing}
& Dataset
& Apache-2.0 \\

MMW benchmark~\cite{piet2025markmywords}
& Dataset / Benchmark
& Apache-2.0 \\

Qwen3 (0.6B, 1.7B, 4B, 8B)~\cite{yang2025qwen3}
& Model
& Apache-2.0 \\

Qwen2.5-3B-Instruct~\cite{qwen2025qwen25technicalreport}
& Model
& Qwen Research License \\

Llama3.1-8B-Instruct~\cite{grattafiori2024llama}
& Model
& \makecell[l]{Llama 3.1 Community License;\\ acceptable use policy applies} \\

Llama3.2-3B-Instruct~\cite{grattafiori2024llama}
& Model
& \makecell[l]{Llama 3.2 Community License;\\ acceptable use policy applies} \\

\bottomrule
\end{tabular}
\label{tab:asset_licenses}
\end{table}


\newpage

\end{document}